\documentclass[amsmath,amssymb,14pt]{article}

\usepackage{graphicx, epsfig}

\usepackage{amssymb,amsmath,color}
\usepackage{amsthm}
\usepackage[a4paper, total={6.5in, 10in}]{geometry}

\numberwithin{equation}{section}

\usepackage{subcaption}
\usepackage{caption}
\usepackage[export]{adjustbox}

\usepackage[nottoc]{tocbibind}

\usepackage{dcolumn}
\usepackage{bm}

\newtheorem{mytheo}{Theorem}[section]
\newtheorem{mylemma}{Lemma}[section]

\newtheorem{myproposition}[mytheo]{Proposition}

\newtheorem{myexample}[mytheo]{Example}

\newtheorem{myremark}[mytheo]{Remark}

\newtheorem{mytheo*}{Theorem}
\newtheorem{mylemma*}{Lemma}
\newtheorem{myproposition*}{Proposition}
\newtheorem{mydef*}{Definition}[section]
\newtheorem{myremark*}{Remark}
\newtheorem{myproof*}{Proof}
\newtheorem*{mynotation*}{Notation}
\newtheorem{mycorollary*}{Corollary}
\newtheorem*{theorem*}{Theorem}

\usepackage{xcolor}

\DeclareMathOperator{\Real}{Re}

\usepackage{mathtools}

\usepackage{authblk}

\begin{document}
\title{On  Schr\"{o}dinger Operators \\ Modified by $\delta$ Interactions}

\author[1]{Kaya G\"{u}ven Akba\c{s}}
\author[2]{Fatih Erman}
\author[3]{O. Teoman Turgut}
\affil[1, 3]{Department of Physics, Bo\u{g}azi\c{c}i University, Bebek, 34342, \.{I}stanbul, Turkey}
\affil[2]{Department of Mathematics, \.{I}zmir Institute of Technology, Urla, 35430, \.{I}zmir, Turkey}
\affil[1]{haci.akbas1989@gmail.com}
\affil[2]{fatih.erman@gmail.com}
\affil[3]{turgutte@boun.edu.tr}

\maketitle

\begin{abstract} We study the spectral properties of a Schr\"{o}dinger operator $H_0$ modified by $\delta$ interactions and show explicitly how the poles of the new Green's function are rearranged relative to the poles of original Green's function of $H_0$. We prove that the new bound state energies are interlaced between the old ones, and the  ground state energy is always lowered if the $\delta$ interaction is attractive. We also derive an alternative perturbative method of finding the bound state energies and wave functions under the assumption of a small coupling constant in a somewhat heuristic manner. We further show that these results can be extended to cases in which a renormalization process is required. We consider the possible extensions of our results to the multi center case, to $\delta$ interaction supported on curves, and to the case, where the particle is moving in a compact two-dimensional manifold under the influence of $\delta$ interaction. Finally, the semi-relativistic extension of the last problem has been studied explicitly. \end{abstract}

\section{Introduction} \label{introduction}

We consider a self-adjoint Schr\"{o}dinger operator or Hamiltonian $H_0=-\frac{\hbar^2}{2m} \Delta + V$ defined on some dense domain $D(H_0)$ and assume that $H_0$ admits a discrete or point spectrum $\sigma_d(H)$ and a (purely absolutely) continuous real spectrum $\sigma_c(H)$, which does not overlap with the discrete spectrum (no embedded eigenvalues). We further assume that the discrete spectrum has no condensation point and the spectrum is bounded below. These conditions put some restrictions on the regularity properties of the potentials, but this is the typical situation in most of the quantum mechanical problems with physically reasonable potentials. The spectrum consists of eigenvalues $E_n$ of $H_0$ with finite multiplicities,   
\begin{eqnarray}
    H_0 \phi_n = E_n \phi_n \;,
\end{eqnarray}
where $\phi_n \in L^2$ are the eigenvectors (or eigenfunctions) corresponding to the eigenvalues $E_n$, and the values $\lambda$ in the continuous spectrum are the ``generalized eigenvalues" corresponding to the ``generalized eigenvectors" $\chi_{\lambda}$ (or generalized eigenfunctions), which is outside of $L^2$. Here we omit the degeneracy/multiplicity labels for simplicity. Although {\it we are not aiming for a completely rigorous presentation in this work}, we point out that the proper meaning of the generalized eigenfunctions has been understood in the context of so-called Rigged Hilbert spaces \cite{Bohm, BohmGadella, BerezinShubin} and the equation for the generalized eigenvalues are given by $H_0 \chi_{\lambda} = \lambda \chi_{\lambda}$ in the sense of distributions. One may think of a non-zero tempered distribution $\chi_{\lambda}$ as a generalized eigenfunction associated with the eigenvalue $\lambda$ for $H_0$ if and only if $\langle H_0\chi, \psi \rangle = \langle \chi_{\lambda}, H_0 \psi \rangle = \lambda \langle \chi_{\lambda}, \psi \rangle$ for any infinitely differentiable rapidly decaying functions $\psi$ (Schwartz functions), see e.g., \cite{Appel} for a more elementary discussion. The function $\chi_{\lambda}(x)=\frac{1}{\sqrt{2\pi}} e^{i \lambda x/\hbar}$ is the well-known generalized eigenfunction of the momentum operator $(P \psi)(x)=-i\hbar \frac{d}{dx} \psi(x)$ in $L^2(\mathbb{R})$. Under relatively mild assumptions on the potential $V$ for dimensions $d\leq 3$ case, it is known that the generalized eigenfunctions (of a Schr\"odinger operator) can actually be selected as {\it continuous functions} \cite{Berezanskii, BeliyKovalenkoSemenov}. This is extremely valuable for our computations as we either evaluate them at a point or integrate them over a curve in this work.

According to one of the fundamental assumptions of quantum mechanics, we can formally expand any function $\psi \in L^2(\mathbb{R})$ as (see e.g., page 48 in \cite{GalindoPascual1}):
\begin{eqnarray}
    \psi(x)= \sum_{n} a_n \phi_n(x) + \int_{\Lambda} a(\lambda) \chi_{\lambda}(x) \; d \mu(\lambda) \;, \label{generalizedeigenfunctionexpansion}
\end{eqnarray}
where 
\begin{eqnarray}
    a_n= \int_{\mathbb{R}} \overline{\phi_n(x)} \psi(x)d x \;, \hspace{1cm} a(\lambda)= \int_{\mathbb{R}} \overline{\chi_{\lambda}(x)} \psi(x)d x \;.
\end{eqnarray}
The  domain of integration in equation (\ref{generalizedeigenfunctionexpansion}) is represented by $\Lambda$, which is a parameter space,  a subset of $\mathbb{R}$ or in general of $\mathbb{R}^n$,  constructed from the absolutely continuous spectrum $\sigma_c(H_0)$, see \cite{Maurin} and the relatively recent work \cite{GadellaGomez}. The expression $d\mu(\lambda)$ in the above expansion is indeed a kind of spectral measure associated with the self-adjoint operator $H_0$. Typically, for most  quantum mechanical problems, there is no singular continuous  spectrum, which allows us to write the above measure  as $d \mu(\lambda)=d\lambda$ (by absorbing a possible positive measurable function into the definition of eigenfunctions), instead of a more general one,  if we interpret the integration as a direct integral \cite{GalindoPascual1}. If there is a degeneracy, we need to sum over those indices as well. Fourier transform is actually a formal eigenfunction expansion of a function $\psi$ in terms of the generalized eigenfunctions $\chi_{\lambda}$ of the momentum operator $P$ and the action of $P$ on the function $\psi$ is given by:
%
%$$
\begin{eqnarray}
    (P\psi) (x) = \int_{-\infty}^{\infty} \lambda \left(\int_{-\infty}^{\infty} \psi(\xi) \overline{\chi_{\lambda}(\xi)} d \xi\right) \chi_{\lambda}(x) d \lambda \;.
\end{eqnarray}

The theory of such eigenfunction expansions was first described by  Weyl \cite{Weyl} for ordinary differential  equations, and then developed partly by Titchmarsh \cite{Titchmarsh1} and by Kodaira \cite{Kodaira}. This is extended to the multidimensional case for {\it elliptic} self-adjoint differential operators in \cite{Berezanskii}. In particular, such expansions in terms of generalized eigenfunctions are rigorously constructed, if we impose the uniformly locally square integrability condition on the potential energy in the Schr\"{o}dinger operator $H_0=-\frac{\hbar^2}{2m} \Delta + V$, which is stated as Theorem 3 in \cite{PoerschkeStolzWeidmann}. In this paper, we tacitly assume that these mild conditions for the eigenfunction expansions hold.

Our main interest here is to consider such Schr\"{o}dinger operators $H_0$ modified by a $\delta$ interaction and show explicitly how the eigenvalues and eigenfunctions change using the eigenfunction expansions of the Green's function for such generic $H_0$. The modification of the bound state spectrum  in one dimension and $d$ dimensional radial case, have been studied in the framework of path integrals in an  influential work of Grosche \cite{Grosche} for some exactly solvable potentials $V$. One of the simplest choices for $H_0$ is the well-known harmonic oscillator problem in one dimension (as well as  the radially symmetric extension of it) and the effect of  adding a $\delta$ interaction to harmonic oscillator Hamiltonians has been studied by many authors \cite{Fassari1, Fassari2, Ersan2} and the statistical properties of this system have been worked out in \cite{JankeCheng}. Moreover, its application  to Bose-Einstein condensation has been investigated in \cite{Huncu}. Another simple example is an infinite square well potential modified by $\delta$ interactions in one dimension and studied in \cite{gadella2014infinite}. A more sophisticated model is a one dimensional $V$-shaped quantum well modified by a point potential centered at the origin has been considered in \cite{fassari2018spectroscopy}.

It is well-known that the description of point like $\delta$ interactions in two and three dimensions requires renormalization and they have been studied in \cite{Albeverio2012solvable, AlbeverioKurasov} in a mathematically rigorous way and in \cite{Huang, Jackiw, GosdzinskyTarrach, MeadGodines, ManuelTarrach} as toy models in understanding of some quantum field theoretical concepts. The higher dimensional version of the harmonic oscillator Hamiltonian with $\delta$ potential requires renormalization as well and has been discussed in \cite{Grosche2} and more recently in \cite{AlbeverioFassari1, AlbeverioFassari2, Fassari3Gadella}. A more interesting exactly solvable example, having both discrete and continuous spectrum, is the so-called reflectionless potential \cite{Grosche, LandauLifshitz}, and we wish to show how its spectral properties change under the influence of $\delta$ interaction as a case study.

It turns out that eigenvalues of $H_0$ modified by $\delta$ interactions change according to some algebraic or transcendental equation and all the eigenvalues of $H_0$ disappear unless some set of wave functions vanish at the support of the $\delta$ interaction or if the support of Dirac delta function is chosen to be at the nodes of the wave function of the initial Hamiltonian $H_0$. This result was not completely illustrated for a  generic Hamiltonian $H_0$, but only shown for particular exactly solvable cases \cite{Grosche, Grosche2}. In all these cases, the full Green's function $G$ of the modified system contains the Green's function $G_0$ of $H_0$ added to  another term constructed again from $G_0$.  It is not at all obvious that the poles of $G_0$ cancel with the poles of this additional term (which has $G_0$ appearing in its expression in a nontrivial combination). This cancellation has only been pointed out for $H_0$ being the one dimensional harmonic oscillator Hamiltonian in \cite{Fassari1} and for higher dimensional harmonic and linear potentials in \cite{Fassari3Gadella}.  In this paper, we  prove this explicitly by using the eigenfunction expansion of the full Green's function by taking the generalized eigenfunction expansion into account and studying its pole structure for a general class of potentials. Moreover, we prove that if the support of the $\delta$ interaction is not at the node of the bound state eigenfunction of $H_0$, then the new eigenvalues $E_{k}^{*}$ are interlaced between $E_{k-1}$ and $E_k$ for an attractive $\delta$ interaction. We then develop a perturbative method from a different perspective to compute order by order the new eigenvalues and wave functions (in a somewhat heuristic way) under the assumption of small coupling, then compare the results  with the standard approach. All the results that we have found can easily be extended to the multi center case, and to the case where we need to apply renormalization. Our perturbation method applied to the problems which require renormalization yields different results from the standard approach obtained by replacing the bare coupling constant $\alpha$ with the renormalized one $\alpha_R$, as expected. We further show that similar conclusions can be drawn for the modification by $\delta$ interactions supported on curves in the plane and 
for the problem, for which the particle is intrinsically moving in a compact manifold under the influence of a $\delta$ interaction.

This paper is organized as follows. In Section \ref{section2}, using the eigenfunction expansion of the Green's function we have explicitly shown that the poles of the Green's function of $H_0$ modified by $\delta$ interaction in one dimension cancels out the poles of the Green's function of $H_0$ under some mild conditions and prove that the new bound state energies are interlaced between the old ones and discuss the results with an explicit exactly solvable reflectionless potential. Then, we develop a new perturbative way of finding the bound state energies and the bound state wavefunctions up to the second order and compare them with the classical known results. 
Section \ref{ModificationsbydeltaPotentialsinSingularCases} deals with the extension of the results to the singular case, where the renormalization of the problem is required. Finally, Section \ref{PossibleGeneralizations} is devoted to the possible extensions of the results, e.g., $N$ center case, $\delta$ interaction supported on a curve in the plane, and the case where a particle is moving on a compact manifold interacting with a point center. We generalize our arguments to a (possible) semi-relativistic version of singular interaction. In Appendix A, we prove the sum $\sum_{n=0}^{\infty} \frac{|\phi_n(a)|^2(E+\mu^2)}{(E_n-E)(E_n+\mu^2)}$ is convergent for particular class of Schr\"{o}dinger operators and for free Schr\"{o}dinger operators on two dimensional compact manifolds. In Appendix B, we prove that there is indeed one single parameter in the renormalized theory. In Appendix C, we prove that the sum $\sum_{n=0}^{\infty} { (E+\mu^2)|\phi_n(a)|^2\over (E_n-E)(E_n+\mu^2)}$ is going to $-\infty$ as $E \to -\infty$ for particular class of Schr\"{o}dinger operators and for free Schr\"{o}dinger operators on two dimensional compact manifolds. In Appendix D, we study our main problem here when there is degeneracy.

\section{Modifications by $\delta$ Interactions in Regular Cases} \label{section2}

Throughout the paper, we use some well-known properties of Green's function. For  convenience of the reader, we state them here with our choice of notation. For instance, the integral kernel of the resolvent $R_{H_0}(z)$ for the Schr\"{o}dinger operator $H_0=-\frac{\hbar^2}{2m} \frac{d^2}{d x^2}+V(x)$ or simply Green's function defined by
\begin{eqnarray}
 \left(R_{H_0}(E)\psi\right)(x)= \left(R_0(E)\psi\right)(x) = \left((H_0-E)^{-1}\psi\right)(x)= \int_{\mathbb{R}} G_0(x,y|E)\psi(y) d y   \;,
\end{eqnarray}
can be expressed as the following bilinear expansion
\begin{eqnarray}
    G_0(x,y|E) = \sum_{n=0}^{\infty} \frac{\phi_n(x) \overline{\phi_n(y)}}{E_n-E} + \int_{\Lambda} \frac{\chi_{\lambda}(x) 
    \overline{\chi_{\lambda}(y)}}{\lambda-E} \; d \mu(\lambda) \;. \label{greenfuncexpansion}
\end{eqnarray}
The Green's function $G_0(x,y|E)$ is a square integrable function of $x$ for almost all $y$ and vice versa \cite{ReedSimonv3}. 
The above formulas still hold in higher dimensions as well.

Suppose that we first consider one dimensional problems in which $H_0$ is modified by a $\delta$ function interaction supported at the origin 
\begin{eqnarray}
    H=H_0 - \alpha \delta \;, \label{Hamiltonian}
\end{eqnarray}
where $\alpha \in \mathbb{R}$. There are different ways to make sense of the above formal expression  of $H$. One way is to consider the $\delta$ interaction as a self-adjoint extension of $H_0$. A modern introduction to this subject is the recent book by Gallone and Michelangeli \cite{allesandro} and the classic reference elaborating this point of view is the monograph by Albeverio et all \cite{Albeverio2012solvable}. $\delta$ interactions can also be defined through the strong limit of the resolvent of Hamiltonians with $\delta$ interaction replaced by some scaled function, see \cite{AlbeverioKurasov, DimockRajeev} for the details. The formulation of the problem in this section can be extended to higher dimensions as long as the co-dimension (dimension of the space - dimension of the support of the $\delta$ interaction) is one. We will come back to this issue for possible generalizations of the problem later on.

It is well-known that Green's function for the Hamiltonian (\ref{Hamiltonian}) is given by Krein's type of formula
\begin{eqnarray} \label{GreensfunctionNR}
G(x,y|E)= G_0(x,y|E) + \frac{G_0(x,0|E) G_0(0,y|E)}{\Phi(E)} \;,
\end{eqnarray}
where $G_0(x,y|E)$ is the Green's function for $H_0$ and 
\begin{eqnarray}
 \Phi(E):= \frac{1}{\alpha}-G_0(0,0|E) \;. \label{principalfunction}
\end{eqnarray}
Some of the books in the context of point interactions use the notation $\Gamma$ for our choice $\Phi$. 
Here and subsequently, as emphasized in the introduction, we assume that $H_0=-\frac{\hbar^2}{2m} \frac{d^2}{dx^2}+V(x)$ satisfies some conditions:
\begin{itemize}
    \item $H_0$ is self-adjoint on some dense domain $D(H_0) \subset L^2(\mathbb{R})$, 
    \item Spectrum of $H_0$ is a disjoint union of discrete $\sigma_d(H_0)$ (set of eigenvalues) and (absolutely) continuous spectrum $\sigma_c(H_0)$,
    \item The discrete spectrum has no accumulation point,
    \item For stability, we assume $H_0$ has spectrum  bounded below,
\end{itemize}
and these assumptions are assumed to hold in higher dimensions as well. These conditions on the spectrum put some mild restrictions on the potential $V$. Some of the possible classes of potentials are listed in the classical work of Reed and Simon \cite{ReedSimonv4}.

Then, we claim our first result about the spectral properties of the Hamiltonian (\ref{Hamiltonian}):
\begin{myproposition} \label{Prop1}
Let $\phi_k(x)$ be the bound state wave function of $H_0$ associated with the bound state energy $E_k$. Then, 
the bound state energy $E_{k}^{*}$, when $H_0$ is modified (perturbed) with a $\delta$ function interaction ($-\alpha \delta(x)$), satisfies the equation
\begin{eqnarray}
\Phi(E)=\frac{1}{\alpha}-G_0(0,0|E) = 0 \;, \label{rootsofPhi}
\end{eqnarray}
if $\phi_k(0) \neq 0$ for  this particular $k$. If for this  choice of  $k$ we have $\phi_k(0)=0$, the bound state energy does not change, $E_{k}^{*}=E_k$.  Moreover, the continuous spectrum of the Hamiltonian modified with $\delta$ interaction is the same as that of $H_0$. 
\end{myproposition}

\begin{proof} We first explicitly show how the pole structure of the full Green's function $G(x,y|E)$ is rearranged and the poles of $G_0(x,y|E)$, which explicitly appears as an additive factor in $G(x,y|E)$ actually are cancelled and new poles appear. Note that from  the explicit expression of the full Green's function (\ref{GreensfunctionNR}), one may expect that the poles of the full Green's function may contain the poles of $G_0$ as well as  the zeroes of the function $\Phi(E)$.
Using the eigenfunction expansion of the Green's function (\ref{greenfuncexpansion}) and splitting the term in the summation associated with the isolated simple eigenvalue $E_k$ of $H_0$, we obtain
\begin{eqnarray} & & \hskip-2cm
    G(x,y|E)=  \frac{\phi_k(x) \overline{\phi_k(y)}}{E_k-E}  \left( 1- \left( 1- \frac{(E_k-E)}{|\phi_k(0)|^2} D(\alpha, E) \right)^{-1}\right) + g(x,y|E)+ h(x,y|E) \nonumber \\ & & + \, \frac{(E_k-E)}{(E_k-E)D(\alpha, E) - |\phi_k(0)|^2} \bigg(g(x,0|E) g(0,y|E) + g(x,0|E) h(x,0|E) \nonumber \\ & & \hspace{2cm} + \, g(x,0|E) h(0,y|E) + h(x,0|E) h(0,y|E)  \bigg) \nonumber \\ & & + \, \frac{1}{(E_k-E)D(\alpha, E) - |\phi_k(0)|^2} \bigg(g(x,0|E) \phi_k(0) \overline{\phi_k(y)} + g(0,y|E) \phi_k(x) \overline{\phi_k(0)}  \nonumber \\ & & \hspace{2cm} + \, h(x,0|E) \phi_k(0) \overline{\phi_k(y)} + h(0,y|E) \phi_k(x) \overline{\phi_k(0)}  \bigg) \;,
\label{greenNR2} \end{eqnarray}
where we have defined the following functions, which are regular near $E_k$,  for simplicity
\begin{eqnarray}
    g(x,y|E) & := & \sum_{n \neq k} \frac{\phi_n(x) \overline{\phi_n(y)}}{E_n-E} \;, \\
    h(x,y|E) & := & \int_{\Lambda} \frac{\chi_{\lambda}(x) \overline{\chi_{\lambda}(y)}}{\lambda-E} \; d \mu(\lambda) \;,\\
    D(\alpha, E) & := & \frac{1}{\alpha}-\sum_{n \neq k} \frac{|\phi_n(0)|^2}{E_n-E} - \int_{\Lambda} \frac{|\chi_{\lambda}(0)|^2}{\lambda-E} \; d \mu(\lambda) \;.
\end{eqnarray}
The functions $g, h$ at $x=0$ and $y=0$, and the function $D$ are well-defined thanks to the fact that the generalized eigenfunctions $\chi_{\lambda}$ in here have continuous representatives \cite{Berezanskii, BeliyKovalenkoSemenov}. Except for the first term in equation (\ref{greenNR2}), it is easy to see that all terms are analytic  in a sufficiently small disk around $E=E_k$. If we choose $E$ sufficiently close to $E_k$, i.e., if $\frac{|E_k-E|}{|\phi_k(0)|^2} \left|D(\alpha, E)\right|<1$, the first term in the above equation becomes
\begin{eqnarray}
- \frac{\phi_k(x) \overline{\phi_k(y)}}{|\phi_k(0)|^2}  \left(\frac{1}{\alpha}-\sum_{n \neq k} \frac{|\phi_n(0)|^2}{E_n-E} - \int_{\Lambda} \frac{|\chi_{\lambda}(0)|^2}{\lambda-E} \; d \mu(\lambda) \right) + O(1) \;,
\end{eqnarray}
so that $G(x,y|E)$ is regular near $E=E_k$ as long as $\phi_k(0)\neq 0$. Hence, the pole $E=E_k$ of $G_0(x,y|E)$ is {\it not a pole} of $G(x,y|E)$ if $\phi_k(0)\neq 0$.
Then, the only poles of $G(x,y|E)$ must  come from the zeroes of the function $\Phi(E)$. 

Finally, it follows from Weyl's theorem \cite{ReedSimonv4} that the continuous spectrum of the problem coincides with the initial Hamiltonian since the difference between the resolvent of the Hamiltonian and the resolvent of the initial Hamiltonian is of finite rank thanks to the explicit formula
\begin{eqnarray}
    R(E)=R_0(E) + (\Phi(E))^{-1} \langle \overline{G(\cdot, 0|E)}, \cdot \rangle G(\cdot, 0|E) \;, \label{resolventexplicitformula}
\end{eqnarray}
defined on its resolvent set. This formula should be understood by its action on a vector $\psi$, i.e.,
\begin{eqnarray}
    (R(E) \psi)(x)=(R_0(E)\psi)(x)+ (\Phi(E))^{-1} \left(\int_{-\infty}^{\infty} G(y, 0|E) \psi(y) \;d y \right) G(x, 0|E) \;.
\end{eqnarray}
The formula (\ref{resolventexplicitformula}) can be seen more naturally in Dirac's bra-ket notation, 
\begin{eqnarray}
    R(E) = R_0(E) + (\Phi(E))^{-1} R_0(E) |0 \rangle \langle 0 | R_0(E) \;.
\end{eqnarray}

\end{proof}
\begin{mylemma} \label{Lemma1}
If $\alpha>0$ (attractive case) and $\phi_k(0) \neq 0$ for some $k \geq 1$, then the new bound state energies $E_{k}^{*}$ are interlaced between the eigenvalues of $H_0$:
\begin{eqnarray}
    E_{k-1} < E_{k}^{*} < E_{k} \;.
\end{eqnarray}
For the ground state ($k=0$), we always have  $E_{0}^{*} < E_0$.

\end{mylemma}
\begin{proof} Since the Green's function $G_0$ has poles at the eigenvalues $E=E_n$ of $H_0$ in the complex $E$ plane and has a branch cut along the generalized eigenvalues of $H_0$, it is differentiable everywhere except at its poles and along the branch cut. Then, by taking the derivative of $G_0(0,0|E)$ with respect to $E$ under the summation and integral sign, we formally obtain
\begin{eqnarray}
    \frac{d G_0(0,0|E)}{d E} = \sum_{n=0}^{\infty} \frac{|\phi_n(0)|^2}{(E_n-E)^2} + \int_{\Lambda} \frac{|\chi_{\lambda}(0)|^2}{(\lambda-E)^2} \; d \mu(\lambda) > 0 \;.
\end{eqnarray}
This implies that $G_0(0,0|E)$ is a monotonically increasing function of $E$ between its poles as well as between the largest eigenvalue below the continuum branch cut and  the infimum of the branch cut (which is typically zero). Moreover, by isolating the $n=k$  term in the sum, it is easy to see that 
\begin{eqnarray}
    \lim_{E \to E_{k}^{-}} \sum_{n=0}^{\infty} \frac{|\phi_n(0)|^2}{E_n-E} + \int_{\Lambda} \frac{|\chi_{\lambda}(0)|^2}{\lambda-E} \; d \mu(\lambda) & = &  \infty \;, \\
    \lim_{E \to E_{k}^{+}} \sum_{n=0}^{\infty} \frac{|\phi_n(0)|^2}{E_n-E} + \int_{\Lambda} \frac{|\chi_{\lambda}(0)|^2}{\lambda-E} \; d \mu(\lambda) & = &  -\infty \;,
\end{eqnarray}
for all $k=0, 1, 2, \cdots$, and by taking the limit under the summation and integral we have
\begin{eqnarray}
    \lim_{E \to -\infty} \sum_{n=0}^{\infty} \frac{|\phi_n(0)|^2}{E_n-E} + \int_{\Lambda} \frac{|\chi_{\lambda}(0)|^2}{\lambda-E} \; d \mu(\lambda) & = &  0^+ \;.
\end{eqnarray}
Then, if $\alpha>0$ it follows from the above results that the roots, say $E_{k}^{*}$, of the equation (\ref{rootsofPhi}) must be located at the points of intersection of $1/\alpha$  and $G_0(0,0|E)$, as shown in Figure \ref{fig:1}.
\begin{figure}[h!!]
    \centering
    \includegraphics[scale=0.3]{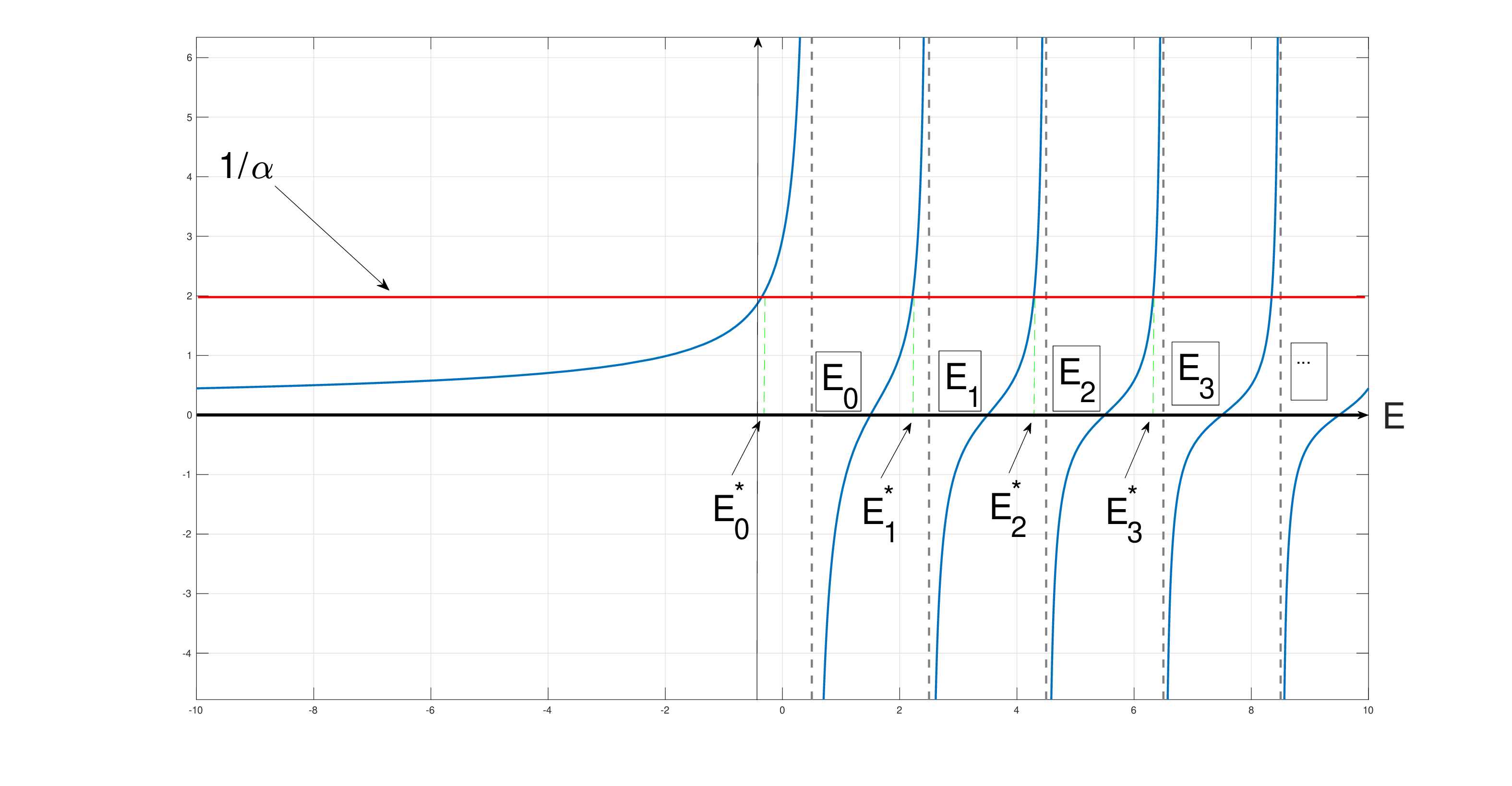}
    \caption{The graph of $1/\alpha$ and a typical behavior of the term $\sum_{n=0}^{\infty} \frac{|\phi_n(0)|^2}{E_n-E} + \int_{\Lambda} \frac{|\chi_{\lambda}(0)|^2}{\lambda-E} \; d \mu(\lambda)$ versus $E$ for $\alpha>0$. The intersection points give the bound state energies. Here $E_0, E_1, E_2, \ldots$ represents the bound state energies of $H_0$ and $E_{0}^{*}, E_{1}^{*}, E_{2}^{*}, \ldots$ represents the bound state energies of the Hamiltonian modified by $\delta$ interaction.}
    \label{fig:1}
\end{figure}
This shows that the new bound state energies $E_{k}^{*}$ are shifted downwards and interlaced between $E_{k-1}$ and $E_k$. Let $E_0$ be the ground state energy of $H_0$. In contrast to excited states, the new ground state energy can be as small as possible by choosing $\alpha$ sufficiently large. By the positivity of the ground state wave functions of Schr\"{o}dinger operators (thanks to Kato's inequality \cite{ReedSimonv4}), having no nodes $\phi_0(0) \neq 0$, it follows from Proposition \ref{Prop1} that new ground state energy $E_{0}^{*}$ is always less than the original one. 

\end{proof}

\begin{myremark}
If $\alpha<0$ (repulsive case), then the bound state energies of $H_0$ are shifted upward and interlaced as $E_k<E_{k}^{*}<E_{k+1}$ if $\phi_k(0) \neq 0$. This can be easily seen by moving the line $1/\alpha$ below the $E$ axis in Figure \ref{fig:1}. Otherwise, bound state energies of $H_0$ do not change.
\end{myremark}

\begin{myremark}
The above statements are still true if we consider the support of $\delta$ potential at $x=a$. In this case,  we have
\begin{eqnarray} \label{GreensfunctionNRdeltaa}
G(x,y|E)= G_0(x,y|E) + \frac{G_0(x,a|E) G_0(a,y|E)}{\frac{1}{\alpha}-G_0(a,a|E)} \;,
\end{eqnarray}
and the bound state energies under the addition of $\delta$ interaction do not change if $a$ is at one of the nodes of the bound state wave function of $H_0$, that is,
$\phi_k(a)=0$. Otherwise, the bound state energies are obtained from $\Phi(E)=\frac{1}{\alpha}-G_0(a,a|E)=0$.
\end{myremark}

\begin{myremark} It is easy to see from Figure \ref{fig:1} that the number of bound states increases by one if we add $\delta$ interaction to $H_0$. An extra pole is created below the ground state energy of $H_0$ and all the other eigenvalues are interlaced. Moreover, the bound state energy of $H_0$ is invariant under a particular configuration of $\delta$ interaction added to $H_0$, that is, if the support of the $\delta$ function is chosen to be at one of the nodes of the bound state wave function associated with the bound state energy $E_k$ of $H_0$, then this bound state energy $E_k$ does not change under the addition (perturbation) of $\delta$ interaction to $H_0$.   
\end{myremark}

\begin{myexample} We can illustrate what we have stated above by working out an exactly solvable case for $H_0$, whose spectrum includes both discrete and continuous parts. Consider the Schr\"{o}dinger operator associated with the reflectionless interaction $V$ given by  
\begin{eqnarray}
    (H_0 \psi)(x)= - \frac{\hbar^2}{2m} \frac{d^2}{dx^2} \psi(x) - \frac{\hbar^2}{2m} \kappa^2 \frac{N(N+1)}{\cosh^2 (\kappa x)} \psi(x)\;, \label{reflectionlesspotentialN}
\end{eqnarray}
where  $N \in \mathbb{N}$. In this case, the discrete spectrum is given by the set of eigenvalues $E_n= - \frac{\hbar^2 \kappa^2}{2 m} \left(N-n\right)^2$, where $n=0,1,2, \ldots N-1$. The number of bound states is finite and equal to $N$, corresponding eigenfunctions are given by 
\begin{eqnarray}
\phi_{n}^{(N)}(x)= \sqrt{\kappa} \left((N-n) \frac{(2N +1-n)!}{n!} \right)^{1/2} P_{N}^{n-N}(\tanh(\kappa x)) \;,
\end{eqnarray}
where $P_{\mu}^{\nu}(x)$ is Legendre function, defined by $P_{\nu}^{\mu}(x)=\frac{1}{\Gamma(1-\mu)} \left( \frac{1+x}{1-x}\right)^{1/2} {_2}F_1(-\nu, \nu+1;1-\mu;\frac{1-x}{2})$ in terms of hypergeometric functions ${_2}F_1$. In this example, the continuous spectrum is the positive real axis, i.e, $\sigma_c(H_0)=[0,\infty)$ and the generalized eigenfunctions are given by $\chi_{k}^{(N)}(x)= \left( \frac{k}{2 \sinh(\pi k/\kappa)}\right)^{1/2} P_{N}^{i k/\kappa}(\tanh(\kappa x))$ with the generalized eigenvalues $E_k= \frac{\hbar^2 k^2}{2m}$, with $k\in \mathbb{R}$.

For simplicity, we choose $N=1$. In this case, there is a single bound state  and the associated normalized eigenfunction, given by
\begin{eqnarray}
    E_0 & = & - \frac{\hbar^2 \kappa^2}{2m} \;, \\ 
    \phi_0(x) & = & \sqrt{\frac{\kappa}{2}} \frac{1}{\cosh(\kappa x)} \;.
\end{eqnarray}
Its generalized eigenfunction for this case ($N=1$) is given by
\begin{eqnarray}
\chi_E(x)= \sqrt{\frac{\kappa}{2 \pi}} e^{i k x} \left(\frac{ik- \kappa \tanh (\kappa x)}{\kappa +i k} \right) \;,
\end{eqnarray}
with the generalized eigenvalues $E=\frac{\hbar^2 k^2}{2m}$. The Green's function of this potential problem in appropriate units has been found in \cite{Grosche} and its explicit formula in accordance with the conventions and units that we use here is given as
\begin{eqnarray} & & \hskip-1cm
    G_0(x,y|E) = \frac{1}{\hbar} \left(\frac{m}{-2 E} \right)^{1/2} \exp \left(-\frac{|x-y| \sqrt{-2m E}}{\hbar}\right) - \frac{\kappa}{2 \cosh(\kappa x) \cosh(\kappa y)(E+\frac{\hbar^2 \kappa^2}{2m})} \nonumber \\ & & \hspace{2cm} \times  \Bigg[ 1- \left(1-\frac{\hbar \kappa}{\sqrt{-2m E}}\right) \cosh \left( \kappa|x-y| \left(1+ \frac{\sqrt{-2m E}}{\kappa \hbar}\right)\right) \Bigg] \;. \label{Greenforreflectionless}
\end{eqnarray}
There is an extra $\hbar$ factor between our convention for the Green's function and the one introduced in \cite{Grosche}.

Since $\phi_0$ has no node, it follows from Proposition \ref{Prop1} that the new eigenvalue due to the addition of $\delta$ interaction at the point $a$ to $H_0$ must satisfy the following equation 
\begin{eqnarray}
    \frac{1}{\alpha} - G(a,a|E) = \frac{1}{\alpha} -\frac{1}{\hbar} \left(\frac{m}{-2 E} \right)^{1/2} + \frac{\kappa}{2 \cosh^2(\kappa a)(E+\frac{\hbar^2 \kappa^2}{2m})}    \left( \frac{\hbar \kappa}{\sqrt{-2m E}}\right) =0 \;.
\end{eqnarray}
From this, one can see that $G_0$ has a simple pole at $E=-\frac{\hbar^2 \kappa^2}{2m}$ and it is singular near the point $E=0$, where the continuous spectrum begins.

It is easy to see from Figure \ref{fig:2} that
\begin{figure}[h!!]
    \centering
    \includegraphics[scale=0.4]{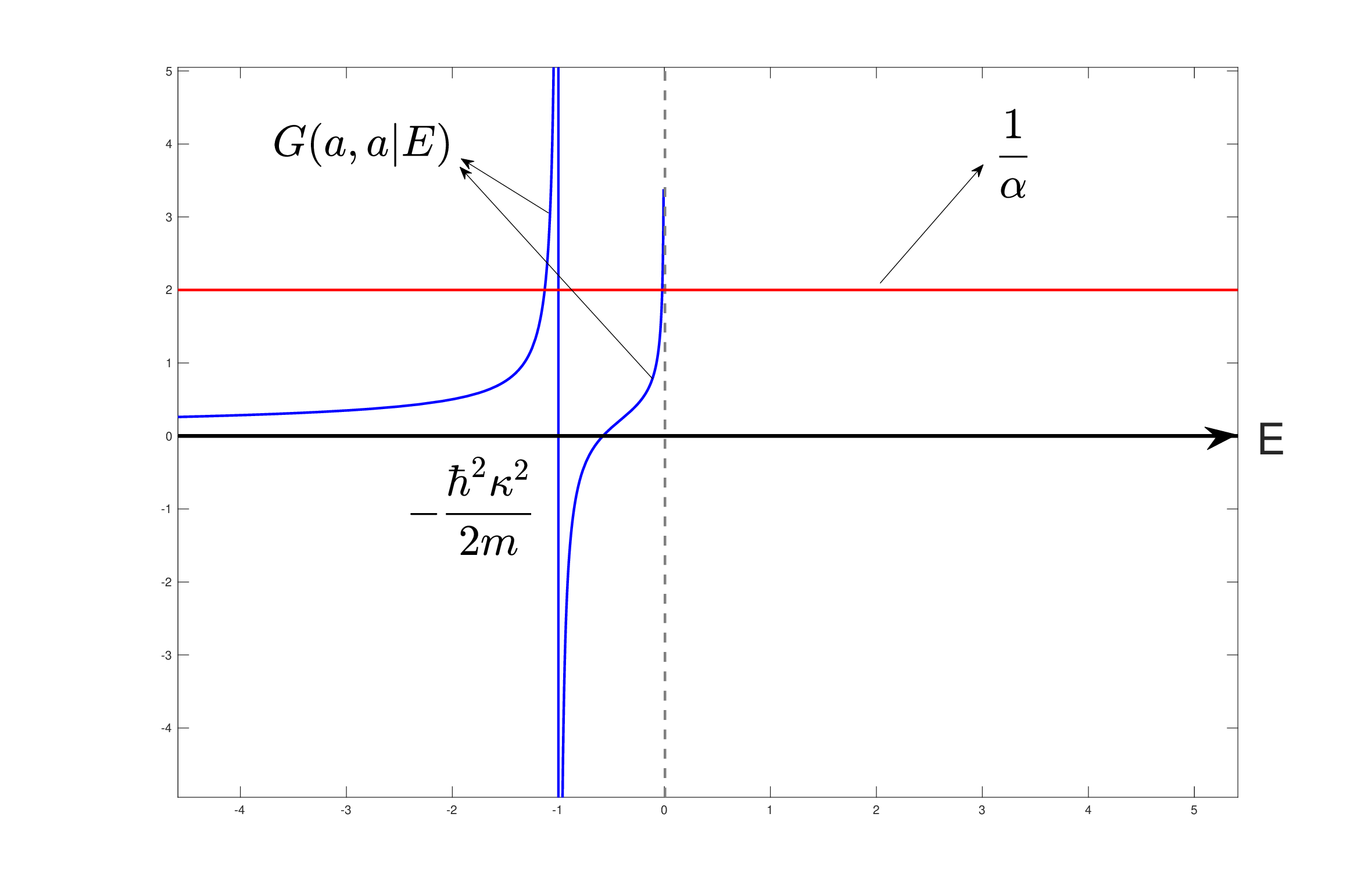}
    \caption{Graphical solution $E$ of the equation $\frac{1}{\alpha}  = G_0(a,a|E)$, where $G_0(a,a|E)= \frac{1}{\hbar} \left(\frac{m}{-2 E} \right)^{1/2} - \frac{\kappa}{2 \cosh^2(\kappa a)(E+\frac{\hbar^2 \kappa^2}{2m})}  \left( \frac{\hbar \kappa}{\sqrt{-2m E}}\right)$. Here we have chosen $\alpha=1/2$, $\hbar=2m=1$, $\kappa=1$, and $a=1$. Here the vertical asymptotes correspond to the bound state energies of the initial Hamiltonian $H_0$ or the point where the continuous spectrum of $H_0$ begins.}
    \label{fig:2}
\end{figure}
as long as $a \neq 0$, we have always two bound state energies, one above $-\frac{\hbar^2 \kappa^2}{2m}$ and one below $-\frac{\hbar^2 \kappa^2}{2m}$, as expected by our above analysis. It is important to notice that the Green function $G_0(a,a|E)$ blows up near its branch point at $E=0$ (where the continuous spectrum begins), which ensures the existence of two new bound states under the addition of $\delta$ interaction to the original problem. However, if $a=0$, then $G_0(0,0|E)$ vanishes at $E=0$. This implies that we have only one root and we have a single bound state energy below $-\frac{\hbar^2 \kappa^2}{2m}$. This is not a generic case as explained above and shown in Figure \ref{fig:1}. This case can be explained by the following symmetry argument. Initially, we have a reflectionless potential having $\mathbb{Z}_2$ symmetry, and it has  a single bound state with even parity. If we include $\delta$ interaction to this potential, we still keep $\mathbb{Z}_2$ symmetry if $a=0$. Then, the excited state must have an odd parity since the systems having such symmetries must have a definite symmetry and the excited state must be orthogonal to the ground state having even parity by the positivity \cite{ReedSimonv4}. This implies that the excited state wave function must vanish at $x=0$, which removes the $\delta$ interaction term in the formal Hamiltonian. However, the initial system $H_0$ has only a single bound state, namely the ground state, so the new system modified by $\delta$ interaction must have a single bound state energy level if $a=0$.

When $a\neq 0$, the bound state energies can be found explicitly and given by 
\begin{eqnarray}
    E_{0}^{*} & = & - \left(\sqrt{\frac{\hbar^2 \kappa^2}{2m}+ \frac{m \alpha^2}{8 \hbar^2}} + \sqrt{\frac{m \alpha^2}{8 \hbar^2}} \right)^2  \;, \\ 
    E_{1}^{*} & = & - \left(\sqrt{\frac{\hbar^2 \kappa^2}{2m}+ \frac{m \alpha^2}{8 \hbar^2}} - \sqrt{\frac{m \alpha^2}{8 \hbar^2}} \right)^2 \;.
\end{eqnarray}

\end{myexample}

\subsection{A Different Perspective for the Perturbative Estimates on the Bound State Energies}
\label{ADifferentPerspectiveforthePerturbativeEstimatesontheBoundStateEnergies}

We now develop a perturbative approximation in finding the eigenvalues of $H_0-\alpha \delta$ with a different perspective. In the standard perturbation theory \cite{LandauLifshitz, GalindoPascual2}, we basically start with the so-called Rayleigh-Schr\"{o}dinger series expansion for the bound state energies and wave functions and then substitute these into Schr\"{o}dinger equation and solve each expansion term recursively (which has been  discussed more rigorously in \cite{ReedSimonv4, Kato}). 

Here, we  follow a somewhat more heuristic approach and assume that $\alpha << 1$ and $\phi_k(0)\neq 0$. According to the previous analysis, new bound state energies will shift due to the addition of $\delta$ interaction to $H_0$. Let $E=E_{k}^* = E_k + \delta E_k$, where $E_k$ are the bound state energies of $H_0$ and $\delta E_k$ is the change in the bound state energy $E_k$. Then, according to Proposition \ref{Prop1} the new bound state energies $E_{k}^*$ are given by the solutions of 
\begin{eqnarray} \hskip-1cm
    & & \frac{1}{\alpha} - \sum_{n=0}^{\infty} \frac{|\phi_n(0)|^2}{E_n-(E_k+\delta E_k)} - \int_{\Lambda} \frac{|\chi_{\lambda}(0)|^2}{\lambda-(E_k+\delta E_k)} \; d \mu(\lambda) \nonumber \\ & & \hspace{1cm} = \frac{1}{\alpha} - \sum_{n\neq k} \frac{|\phi_n(0)|^2}{E_n-(E_k+\delta E_k)} + \frac{|\phi_k(0)|^2}{\delta E_k} - \int_{\Lambda} \frac{|\chi_{\lambda}(0)|^2}{\lambda-(E_k+\delta E_k)} \; d \mu(\lambda) = 0 \;.
\end{eqnarray}
This is an exact equation determining the bound state energies but it involves the eigenfunctions at $x=0$. If we  expand the terms in the summation and integral in the powers of $\delta E_k$ and multiply the equation by $\alpha \, \delta E_k$, we get
\begin{eqnarray} & & 
\delta E_k - \alpha \, \delta E_k \sum_{n\neq k} \frac{|\phi_n(0)|^2}{E_n-E_k} \left( 1+ \frac{\delta E_k}{E_n-E_k} +O(\delta E_{k}^{2})\right) 
\nonumber \\ & & - \alpha \, \delta E_k \int_{\Lambda} \frac{|\chi_{\lambda}(0)|^2}{\lambda-E_k} \; \left( 1+ \frac{\delta E_k}{\lambda-E_k} +O(\delta E_{k}^{2})\right)  d \mu(\lambda)  + \alpha |\phi_k(0)|^2 + O(\delta E_{k}^3) = 0 \;. \label{perturbativebsnonrenormcase}
\end{eqnarray}
Let us assume that $\delta E_k$ has a  power series in $\alpha$, that is, $\delta E_k =  E_{k}^{(1)} +   E_{k}^{(2)} + \cdots$, where $E_{k}^{(n)}$ corresponds to the change in the bound state energy of order $\alpha^n$. Then, solving $E_{k}^{(1)}$ and $E_{k}^{(2)}$ term by term, we obtain
\begin{eqnarray}
     E_{k}^{(1)} & = &  - \alpha |\phi_k(0)|^2 \;, \label{NRenergyperturbation1} \\   E_{k}^{(2)}  & = &  - \alpha^2 |\phi_k(0)|^2 \left( \sum_{n\neq k} \frac{|\phi_n(0)|^2}{E_n-E_k} + \int_{\Lambda} \frac{|\chi_{\lambda}(0)|^2}{\lambda-E_k}  d \mu(\lambda)  \right) \;,
\label{NRenergyperturbation2}\end{eqnarray}
which are consistent with classical first and second order formulas in non-degenerate perturbation theory \cite{GalindoPascual1, LandauLifshitz}, given by
\begin{eqnarray}
     E_{k}^{(1)} & = &  \langle \phi_k, -\alpha \delta \phi_k \rangle \label{perturbationenergyformula1}  \\    E_{k}^{(2)} & = &  - \sum_{n\neq k} \frac{|\langle \phi_k, -\alpha \delta \phi_n \rangle|^2}{E_n-E_k}
    - \int_{\Lambda} \frac{|\langle \chi_{\lambda}, -\alpha \delta \phi_k \rangle|^2}{\lambda-E_k}  d \mu(\lambda) \;.
\label{perturbationenergyformula2}\end{eqnarray}
The important point to note here is that we are looking for the solution $\delta E_k$ of (\ref{perturbativebsnonrenormcase}) as the formal power series in $\alpha$ 
by assuming both $\alpha$ and $\delta E_k$ are small so that it is sufficient to use regular perturbation theory \cite{Murdock}. Let us summarize what we have found as the following Proposition:
\begin{myproposition}
Let $\phi_k$ be the eigenfunctions associated with the eigenvalues $E_k$ of $H_0$ and $\chi_{\lambda}$ be the generalized eigenfunctions of $H_0$. Then, under the addition of $\delta$ interaction with the coupling constant $-\alpha$, the change in the bound state energies up to the first and second order in $\alpha$ are formally given by the equations (\ref{NRenergyperturbation1}) and (\ref{NRenergyperturbation2}).   
\end{myproposition}

\subsection{A Different Perspective for the Perturbative Estimates on the Bound State Wave Function}
\label{ADifferentPerspectiveforthePerturbativeEstimatesontheBoundStateWaveFunction}

Using a contour integral of the resolvent $R(E)=(H-E)^{-1}$ around each simple eigenvalue $E_{k}^{*}$, we can find the projection operator onto the eigenspace associated with the eigenvalue $E_{k}^{*}$,
\begin{eqnarray}
    \mathbb{P}_k = - \frac{1}{2\pi i} \oint_{\Gamma_k} R(E) \; d E \;,
\end{eqnarray}
where $\Gamma_k$ is the closed contour around each simple pole $E_{k}^{*}$, or equivalently 
\begin{eqnarray}
    \psi_k(x) \overline{\psi_k(y)} = - \frac{1}{2\pi i} \oint_{\Gamma_k} G(x,y|E) \; d E \;.
\end{eqnarray}
From the explicit expression of the Green's function (\ref{GreensfunctionNR}) and residue theorem, we obtain 
\begin{eqnarray}
    \psi_k(x)=  \frac{G_0(x,0|E_{k}^{*})}{ \left(\frac{d G_0(0,0|E)}{d E}\bigg|_{E=E_{k}^{*}}\right)^{1/2}} \;. \label{boundstatewavefunctionexact}
\end{eqnarray}
Using the eigenfunction expansion of $G_0$, and $E_{k}^{*}=E_k + \delta E_k$, we get 
\begin{eqnarray}
 \psi_k(x)= \frac{\left( \sum_{n=0}^{\infty} \frac{\phi_n(x) \overline{\phi_n(0)}}{E_n - (E_k +\delta E_k)} + \int_{\Lambda} \frac{\chi_{\lambda}(x) \overline{\chi_{\lambda}(0)}}{\lambda-(E_k+\delta E_k)}  d \mu(\lambda)  \right)}{\left( \sum_{n=0}^{\infty} \frac{|\phi_n(0)|^2}{\left(E_n - (E_k +\delta E_k)
 \right)^2} + \int_{\Lambda} \frac{|\chi_{\lambda}(0)|^2}{(\lambda-(E_k+\delta E_k))^2}  d \mu(\lambda) \right)^{1/2}} \;.    
\end{eqnarray}
If we split the $n=k$  term in the sums and integrals, and then expand the above terms in powers of $\delta E_k$, 
then the wave function becomes
\begin{eqnarray} & & \hskip-1cm
 \psi_k(x) =  \left[1 + \frac{\delta E_{k}^2}{|\phi_k(0)|^2} \sum_{n\neq k} 
   \frac{|\phi_n(0)|^2}{(E_n-E_{k})^2} + \frac{\delta E_{k}^2}{|\phi_k(0)|^2} \int_{\Lambda} \frac{|\chi_{\lambda}(0)|^2}{(\lambda-E_k)^2}  d \mu(\lambda) + O(\delta E_{k}^3) \right]^{-1/2}   \nonumber \\ & & \hspace{2cm} \times \Bigg( - \frac{\phi_k(x) \overline{\phi_k(0)}}{|\phi_k(0)|} +  
    \left( \sum_{n \neq k} \frac{\phi_n(x) \overline{\phi_n(0)}}{(E_n - E_k)|\phi_k(0)|} \left(\delta E_k + \frac{\delta E_{k}^2}{(E_n-E_k)} + O(\delta E_{k}^3) \right) \right) \nonumber \\ & & + \left(\int_{\Lambda} \frac{\chi_{\lambda}(x) \overline{\chi_{\lambda}(0)}}{(\lambda-E_k)|\phi_k(0)|}  \left(\delta E_k + \frac{\delta E_{k}^2}{(\lambda-E_k)} + O(\delta E_{k}^3) \right)  d \mu(\lambda) \right) + O(\delta E_{k}^{3})\;,
\end{eqnarray}
where we have assumed $|\frac{\delta E_k}{(E_n-E_k)}|<1$ and $|\frac{\delta E_k}{(\lambda-E_k)}|<1$. Keeping the terms in $\delta E_k$ up to the second order, we obtain
\begin{eqnarray} & & \hskip-1cm
 \psi_k(x) = \phi_k(x) e^{-i\theta_k+i \pi} + \sum_{n \neq k} \frac{\phi_n(x) \overline{\phi_n(0)}}{|\phi_k(0)|} \frac{E_{k}^{(1)}}{E_n-E_k} + \int_{\Lambda} \frac{\chi_{\lambda}(x) \overline{\chi_{\lambda}(0)}}{|\phi_k(0)|} \frac{E_{k}^{(1)}}{\lambda- E_k}  d \mu(\lambda) \nonumber \\ & & \hspace{2cm} -  
   \frac{1}{2} \phi_k(x) e^{-i\theta_k+i \pi} \left(\sum_{n\neq k} 
   \frac{|\phi_n(0)|^2}{|\phi_k(0)|)^2} \frac{(E_{k}^{(1)})^2}{(E_n-E_k)^2} + \int_{\Lambda} \frac{|\chi_{\lambda}(0)|^2}{|\phi_k(0)|^2} \frac{(E_{k}^{(1)})^2}{(\lambda-E_k)^2} d \mu(\lambda) \right) 
    \nonumber \\ & &   +  \sum_{n \neq k} \frac{\phi_n(x) \overline{\phi_n(0)}}{(E_n - E_k)|\phi_k(0)|} \left(E_{k}^{(2)} + \frac{(E_{k}^{(1)})^2}{(E_n-E_k)} \right) 
   \nonumber \\ & &   \hspace{2cm} + \left(\int_{\Lambda} \frac{\chi_{\lambda}(x) \overline{\chi_{\lambda}(0)}}{(\lambda-E_k)|\phi_k(0)|}  \left(E_{k}^{(2)} + \frac{(E_{k}^{(1)})^{2}}{(\lambda-E_k)} \right)  d \mu(\lambda) \right) + O(\delta E_{k}^{3})\;, \label{wavefunctionnonrenormcase}
\end{eqnarray}
where $\phi_k(0)=|\phi_k(0)| e^{i \theta_k}$. Using the first order and second order results (\ref{perturbationenergyformula1}) and (\ref{perturbationenergyformula2}) for the bound state energies, we finally obtain the bound state wave function for each order:
\begin{eqnarray} \hskip-1cm
 \psi_{k}^{(0)}(x) & = &  e^{-i \theta_k + i \pi} \phi_{k}(x) \label{zeroorderperturbationwavefunction} \\  \psi_{k}^{(1)}(x) & = & \alpha \; \phi_k(0) e^{-i \theta_k + i \pi} \sum_{n \neq k} \frac{\phi_n(x) \overline{\phi_n(0)}}{E_n-E_k}  + \alpha \; \phi_k(0) e^{-i \theta_k + i \pi} \int_{\Lambda} \frac{\chi_{\lambda}(x) \overline{\chi_{\lambda}(0)}}{\lambda- E_k}  d \mu(\lambda)  \label{firstorderperturbationwavefunction} \\  \psi_{k}^{(2)}(x) & = & - \frac{\alpha^2}{2} \phi_k(x) e^{-i\theta_k+i \pi} |\phi_k(0)|^2 \left(\sum_{n\neq k} \frac{|\phi_n(0)|^2}{(E_n-E_k)^2} + \int_{\Lambda}  \frac{|\chi_{\lambda}(0)|^2}{(\lambda-E_k)^2} d \mu(\lambda) \right) 
    \nonumber \\ & &  \hskip-1cm +  \alpha^2 \phi_k(0) e^{-i \theta_k +i \pi}  \sum_{n \neq k} \frac{\phi_n(x) \overline{\phi_n(0)}}{(E_n - E_k)} \Bigg( \sum_{m\neq k} \frac{|\phi_m(0)|^2}{E_m-E_k} + \int_{\Lambda} \frac{|\chi_{\lambda}(0)|^2}{\lambda-E_k}  d \mu(\lambda)   - \frac{ |\phi_k(0)|^2}{(E_n-E_k)} \Bigg)  \nonumber \\ & &   \hskip-2.5cm
   +  \alpha^2  \phi_k(0) e^{-i \theta_k + i \pi}  \int_{\Lambda} \frac{\chi_{\lambda}(x) \overline{\chi_{\lambda}(0)}}{(\lambda-E_k)}   \Bigg(\sum_{m\neq k} \frac{|\phi_m(0)|^2}{E_m-E_k} + \int_{\Lambda} \frac{|\chi_{\lambda'}(0)|^2}{\lambda'-E_k}  d \mu(\lambda')  - \frac{|\phi_k(0)|^2}{\lambda -E_k}  \Bigg)   d \mu(\lambda) \;. \label{secondorderperturbationwavefunction}
\end{eqnarray}

\begin{myremark}
According to standard regular perturbation theory results, the normalized wave functions for the bound states up to second order in $\alpha$ are given by
\begin{eqnarray}
 \psi_{k}^{(0)}(x) & = &  \phi_{k}(x) \\  \psi_{k}^{(1)}(x) & = & \sum_{n \neq k} \frac{\langle \phi_n, -\alpha \delta \phi_k \rangle}{E_k-E_n} \phi_n(x)  + \int_{\Lambda} \frac{\langle \chi_{\lambda}, -\alpha \delta \phi_k \rangle}{\lambda- E_k} \chi_{\lambda}(x) d \mu(\lambda)  \\  \psi_{k}^{(2)}(x) & = &
 \sum_{n \neq k} \sum_{m \neq k} \frac{\langle \phi_n, - \alpha \delta \phi_m \rangle \langle \phi_m, - \alpha \delta \phi_k \rangle}{(E_k-E_n)(E_k-E_m)} \phi_n(x)
 - \sum_{n \neq k} \frac{\langle \phi_n, - \alpha \delta \phi_k \rangle \langle \phi_k, -\alpha \delta \phi_k \rangle}{(E_k-E_n)^2} \phi_n(x) \nonumber \\ & - & \frac{1}{2} \sum_{n \neq k} \frac{\langle \phi_k, -\alpha \delta \phi_n \rangle \langle \phi_n, -\alpha \delta \phi_k \rangle}{(E_k-E_n)^2} \phi_k(x) \nonumber \\ & + &  \int_{\Lambda} \left( \sum_{n \neq k} \frac{\langle \phi_n, - \alpha \delta \chi_{\lambda} \rangle \langle \chi_{\lambda}, - \alpha \delta \phi_k \rangle}{(E_k-E_n)(E_k-\lambda)} \phi_n(x) \right) d \mu(\lambda) \nonumber \\ & + &  \int_{\Lambda} \left( \sum_{m \neq k} \frac{\langle \chi_{\lambda}, - \alpha \delta \phi_m \rangle \langle \phi_m, - \alpha \delta \phi_k \rangle}{(E_k-E_m)(E_k-\lambda)} \chi_{\lambda}(x) \right) d \mu(\lambda) \nonumber \\ & + &  \int_{\Lambda} \left(\int_{\Lambda} \frac{\langle \chi_{\lambda}, - \alpha \delta \chi_{\lambda'} \rangle \langle \chi_{\lambda'}, - \alpha \delta \phi_k \rangle}{(E_k-\lambda)(E_k-\lambda')} \chi_{\lambda}(x) d \mu(\lambda') \right) d \mu(\lambda) \nonumber \\ & & \hskip-2.3cm -  \int_{\Lambda}  \frac{\langle \chi_{\lambda}, - \alpha \delta \phi_k \rangle \langle \phi_k, -\alpha \delta \phi_k \rangle}{(E_k-\lambda)^2} d \mu(\lambda) \; \chi_{\lambda}(x) - \frac{1}{2} \int_{\Lambda}  \frac{\langle \phi_k, - \alpha \delta \chi_{\lambda} \rangle \langle \chi_{\lambda}, -\alpha \delta \phi_k \rangle}{(E_k-\lambda)^2} d \mu(\lambda) \; \phi_k(x) \;.
\end{eqnarray}
These are completely consistent with our results up to a phase factor $e^{-i \theta_k + i \pi}$. (Second order perturbation result for the normalized wave function only in the discrete case has been given in \cite{Solyom}, but clearly it can be generalized as we have stated here). 
\end{myremark}
Let us summarize our findings as 
\begin{myproposition} The bound state wave functions $\psi_k(x)$ for the Hamiltonian $H_0$ modified by the $\delta$ interaction are given by (\ref{boundstatewavefunctionexact}) in terms of the associated bound state energies $E_{k}^{*}$, which are the solutions of $\frac{1}{\alpha}-G_0(0,0|E)=0$. The expansion of the bound state wave function $\psi_k(x)$ in terms of the coupling constants $\alpha$ up to second order are given by (\ref{zeroorderperturbationwavefunction}), (\ref{firstorderperturbationwavefunction}), and (\ref{secondorderperturbationwavefunction}).      
\end{myproposition}

\section{Singular Modifications by $\delta$ Interactions }
\label{ModificationsbydeltaPotentialsinSingularCases}

When the co-dimension (dimension of the space - dimension of the support of the $\delta$ interaction) is greater than one (e.g., point $\delta$ interaction in two and three dimensions, $\delta$ interaction supported by a curve in three dimensions), we need to define $\delta$ interaction by a renormalization procedure. The reason for this is essentially based on the singularity of the Green's function for free Hamiltonians $H_0$ in two and three dimensions. The history of the subject is rather rich and there has been a vast amount of material in the physics literature, see e.g., \cite{Huang, Jackiw, GosdzinskyTarrach, MeadGodines, ManuelTarrach}. The subject has been also discussed  thoroughly from a more mathematical point of view in classic monographs \cite{Albeverio2012solvable, AlbeverioKurasov} as well as in a more recent work \cite{allesandro}. Here we now consider the Schr\"{o}dinger operators $H_0$ with the same assumption discussed before and we study the spectrum of $H_0$ modified by  $\delta$ interactions.

It is useful to express the Green's function $G_0$ in terms of the heat kernel $K_t(x,y)$ associated with the operator $H_0=-\frac{\hbar^2}{2m} \Delta + V$, given by 
\begin{eqnarray}
    G_0(x,y|E)= \int_{0}^{\infty} K_t(x,y) e^{t E} dt \;, \label{intrepgreen}
\end{eqnarray}
where $\Real(E)<0$ and $H_0 K_t(x,y)= \frac{\partial}{\partial t}K_t(x,y)$ and defined for other values of $E$ in the complex $E$ plane through analytical continuation. We note  that the first term in the short time asymptotic expansion of the diagonal heat kernel for any self-adjoint elliptic second order differential operator \cite{Gilkey} in $d$ dimensions, is given by
\begin{eqnarray}
 K_t(x,x) \sim  t^{-d/2} \;. \end{eqnarray}
This gives rise to the divergence near $t=0$ in the diagonal part $G_0(x,x|E)$:  
\begin{eqnarray}
    \int_{0}^{\infty} \frac{e^{-t|E|}}{t^{d/2}} \; dt \;,
\end{eqnarray}
for $d=2,3$. In order to make sense of such singular interactions, one must first introduce a cut-off $\epsilon>0$ and regularize the Hamiltonian. This could be done  replacing the $\delta$ interaction by the heat kernel $K_{\epsilon/2}(x,0)$, which converges to $\delta(x)$ as $\epsilon \to 0$ in the distributional sense. Then, we make the coupling constant $\alpha$ dependent on the cut-off in such a way that the regularized Hamiltonian has a non-trivial limit (in the norm resolvent sense) as we remove the cut-off. A natural choice for the coupling constant is given by  
\begin{eqnarray}
 \frac{1}{\alpha(\epsilon)} = \frac{1}{\alpha_R(M)} + \int_{\epsilon}^{\infty} K_t(0,0) e^{t M} d t \;,
\end{eqnarray}
where $M$ is the renormalization scale and could be
eliminated in favor of a physical parameter by imposing the
renormalization condition. For instance, $M$, and therefore $\alpha_R$, can be related to the bound state
energy of the particle moving under the addition of $\delta$ interaction to $H_0$, say $-\mu^2$ (this requirement, in general, leads to a flow in the space of coupling constants $\alpha_R$ \cite{pointinteractionsonmanifolds2}).  A special choice is to set ${1\over\alpha_R}=0$ while demanding $M=-\mu^2$, which is particularly convenient for bound state calculations. For our purposes, we use $\alpha_R$ and moreover set $M=-\mu^2$ (thinking of a bound state below $E_0$) for clarity of notation (The single parameter dependence of the model has been explicitly shown in Appendix B).

Then, taking the formal limit as $\epsilon \to 0$, we obtain the integral kernel of the resolvent or Green's function, given by 
\begin{eqnarray}
G(x,y|E)= G_0(x,y|E) + \frac{G_0(x,0|E) G_0(0,y|E)}{\frac{1}{\alpha_R}- \sum_{n=0}^{\infty} \frac{|\phi_n(0)|^2(E+\mu^2)}{(E_n-E)(E_n+\mu^2)}- \int_{\Lambda} \frac{|\chi_{\lambda}(0)|^2 (E+\mu^2)}{(\lambda+\mu^2)(\lambda-E)} \; d \mu(\lambda)} \;.
\end{eqnarray}
This is exactly in the same form as the Krein's formula given by (\ref{GreensfunctionNR}), where the form of the function $\Phi$ here is given by 
\begin{eqnarray}
   \Phi(E) & = & \frac{1}{\alpha_R} + \int_{0}^{\infty} K_t(0,0) \left( e^{-t\mu^2}- e^{t E} \right) \; d t \nonumber \\ & = &  \frac{1}{\alpha_R}- \sum_{n=0}^{\infty} \frac{|\phi_n(0)|^2(E+\mu^2)}{(E_n-E)(E_n+\mu^2)}- \int_{\Lambda} \frac{|\chi_{\lambda}(0)|^2 (E+\mu^2)}{(\lambda+\mu^2)(\lambda-E)} \; d \mu(\lambda) \;.
\end{eqnarray}
Here we have used the eigenfunction expansion of the heat kernel. Such  specific examples are examined in \cite{Grosche2} in the context of path integrals in two and three dimensions. However, there is no  explicit derivation in this work,  showing that the poles of the free resolvent are cancelled in the final expression.

Following the same line of argument developed for the one-dimensional case, we can now explicitly show how the poles of the Hamiltonian $H_0$ disappear with the addition of $\delta$ interaction under the same assumptions about the spectrum of $H_0$ as in the one dimensional case.

To simplify our arguments, we assume a purely discrete spectrum, the generalization to include a continuum with generalized eigenfunctions is fairly straightforward. We again split the term in the eigenfunction expansion of the Green's function associated with the isolated simple eigenvalue $E_k$ of $H_0$:
\begin{eqnarray} & & \hskip-1cm
    G(x,y|E)= \sum_{n \neq k} \frac{\phi_n(x) \overline{\phi_n(y)}}{E_n-E} + \frac{\phi_k(x) \overline{\phi_k(y)}}{E_k-E} + \frac{ \left(\sum_{n \neq k} \frac{\phi_n(x) \overline{\phi_n(0)}}{E_n-E} \right) \left( \sum_{n \neq k} \frac{\phi_n(0) \overline{\phi_n(y)}}{E_n-E}\right)}{\frac{1}{\alpha_R}-\sum_{n \neq k} \frac{|\phi_n(0)|^2(E+\mu^2)}{(E_n-E)(E_n+\mu^2)} - \frac{|\phi_k(0)|^2(E+\mu^2)}{(E_k-E)(E_k+\mu^2)}}  \nonumber \\ & & \hspace{1cm}+  \frac{\left(\sum_{n \neq k} \frac{\phi_n(x) \overline{\phi_n(0)}}{E_n-E} \right) \left(\frac{\phi_k(0) \overline{\phi_k(y)}}{E_k-E}\right)}{\frac{1}{\alpha_R}-\sum_{n \neq k} \frac{|\phi_n(0)|^2(E+\mu^2)}{(E_n-E)(E_n+\mu^2)} - \frac{|\phi_k(0)|^2(E+\mu^2)}{(E_k-E)(E_k+\mu^2)}} 
    \nonumber \\ & & \hspace{2cm} + \frac{ \left(\frac{\phi_k(x) \overline{\phi_k(0)}}{E_k-E} \right) \left( \sum_{n \neq k} \frac{\phi_n(0) \overline{\phi_n(y)}}{E_n-E}\right)}{\frac{1}{\alpha_R}-\sum_{n \neq k} \frac{|\phi_n(0)|^2(E+\mu^2)}{(E_n-E)(E_n+\mu^2)} - \frac{|\phi_k(0)|^2(E+\mu^2)}{(E_k-E)(E_k+\mu^2)}} \nonumber \\ & & \hspace{3cm} + \frac{ \left(\frac{\phi_k(x) \overline{\phi_k(0)}}{E_k-E} \right) \left( \frac{\phi_k(0) \overline{\phi_k(y)}}{E_k-E}\right)}{\frac{1}{\alpha_R}-\sum_{n \neq k} \frac{|\phi_n(0)|^2(E+\mu^2)}{(E_n-E)(E_n+\mu^2)} - \frac{|\phi_k(0)|^2(E+\mu^2)}{(E_k-E)(E_k+\mu^2)}} \;.
\end{eqnarray}
If we combine the second and the last term in the above expression, we obtain
\begin{eqnarray} & & \hskip-1cm
    G(x,y|E)= \frac{\phi_k(x) \overline{\phi_k(y)}}{E_k-E}  \left(1- \left(1- \frac{(E_k-E)}{|\phi_k(0)|^2} \left(\frac{1}{\alpha_R}-\sum_{n \neq k} \frac{|\phi_n(0)|^2}{E_n-E} + \frac{|\phi_k(0)|^2}{E_k+\mu^2}\right)\right)^{-1}\right) \nonumber \\ & & \hspace{1cm} + \sum_{n \neq k} \frac{\phi_n(x) \overline{\phi_n(y)}}{E_n-E} + (E_k-E) \frac{ \left(\sum_{n \neq k} \frac{\phi_n(x) \overline{\phi_n(0)}}{E_n-E} \right) \left( \sum_{n \neq k} \frac{\phi_n(0) \overline{\phi_n(y)}}{E_n-E}\right)}{(E_k-E) \left(\frac{1}{\alpha_R}-\sum_{n \neq k} \frac{|\phi_n(0)|^2(E+\mu^2)}{(E_n-E)(E_n+\mu^2)}\right) - \frac{|\phi_k(0)|^2(E+\mu^2)}{(E_k+\mu^2)}}  \nonumber \\ & & \hspace{2cm} +  \frac{\left(\sum_{n \neq k} \frac{\phi_n(x) \overline{\phi_n(0)}}{E_n-E} \right) \left(\phi_k(0) \overline{\phi_k(y)}\right)}{(E_k-E)\left(\frac{1}{\alpha_R}-\sum_{n \neq k} \frac{|\phi_n(0)|^2(E+\mu^2)}{(E_n-E)(E_n+\mu^2)}\right) - \frac{|\phi_k(0)|^2(E+\mu^2)}{(E_k-E)(E_k+\mu^2)}} 
    \nonumber \\ & & \hspace{3cm} + \frac{ \left(\phi_k(x) \overline{\phi_k(0)} \right) \left( \sum_{n \neq k} \frac{\phi_n(0) \overline{\phi_n(y)}}{E_n-E}\right)}{(E_k-E) \left(\frac{1}{\alpha_R}-\sum_{n \neq k} \frac{|\phi_n(0)|^2(E+\mu^2)}{(E_n-E)(E_n+\mu^2)}\right) - \frac{|\phi_k(0)|^2(E+\mu^2)}{(E_k+\mu^2)}} \;.
\end{eqnarray}
Except for the first term, it is easy to see that all terms are regular near $E=E_k$. For the first term, if we choose $E$ sufficiently close to $E_k$, i.e., if $\frac{|E_k-E|}{|\phi_k(0)|^2} \left|\frac{1}{\alpha_R}-\sum_{n \neq k} \frac{|\phi_n(0)|^2}{E_n-E} + \frac{|\phi_{k}(0)|^2}{E_k + \mu^2} \right|<1$, the first term in the above equation becomes
\begin{eqnarray}
- \frac{\phi_k(x) \overline{\phi_k(y)}}{|\phi_k(0)|^2}  \left(\frac{1}{\alpha_R}-\sum_{n \neq k} \frac{|\phi_n(0)|^2}{E_n-E} + \frac{|\phi_{k}(0)|^2}{E_k + \mu^2} 
\right) + O(|E_k-E|^2)
\end{eqnarray}
so that $G(x,y|E)$ is regular near $E=E_k$ as long as $\phi_k(0)\neq 0$. If we assume the presence of a continuous spectrum, this part of the spectrum does not change by the same reasoning as given before. Hence we have 
\begin{myproposition} Let $\phi_k(x)$ be the bound state wave function of $H_0$ associated with the bound state energy $E_k$.
Then, the bound state energies $E_{k}^{*}$ for $H_0$ modified (perturbed) with attractive $\delta$ interactions satisfy the equation
\begin{eqnarray}
\frac{1}{\alpha_R}- \sum_{n=0}^{\infty} \frac{|\phi_n(0)|^2(E+\mu^2)}{(E_n-E)(E_n+\mu^2)} - \int_{\Lambda} \frac{|\chi_{\lambda}(0)|^2 (E+\mu^2)}{(\lambda-E)(\lambda+\mu^2)} \; d \mu(\lambda)  = 0 \label{boundstatenergyrenorm}\;,
\end{eqnarray}
if $\phi_k(0) \neq 0$ for some this particular $k$. If for this choice of $k$ we have $\phi_k(0)=0$, the bound state energies do not change, $E_{k}^{*}=E_k$.  Moreover, the continuous spectrum of the Hamiltonian modified with $\delta$ interaction is the same as that of $H_0$. 
\end{myproposition}

In the renormalized case, the interlacing of the energy eigenvalues is exactly the same as the one in Lemma \ref{Lemma1}, where renormalization is not required,
and if $\alpha_R<0$, then all the bound state energies are shifted downward even further, which is not something obvious. Moreover $\alpha_R \to 0^-$  corresponds to infinitely strong coupling case, so in a sense  ${1\over \alpha_R}$ is the true coupling constant. This is due to the fact that the Green's function is an increasing function of $E$ and for $E\to -\infty$ we find that  the renormalized sum  goes to $ -\infty$, as we let $E\to -\infty$, hence there is always a solution below $-\mu^2$ (this is where the sum vanishes). We will illustrate this for a manifold in an Appendix C.

\begin{myremark}
Note that these results can be interpreted as a generalization of the well-known Sturm comparison theorems to the singular $\delta$ interactions, it is remarkable that even the renormalized case has this property. 
\end{myremark}

Following the same line of arguments as in the regular problem, we now develop a perturbative approximation to the eigenvalues for $H_0$ by the addition of the singular $\delta$ interaction. We assume that $\alpha_R << 1$. Let $E=E_{k}^* = E_k + \delta E_k$, where $E_k$ are the bound state energies of $H_0$. Then, the poles of the Green's function $G(x,y|E)$ are given by the solutions of 
\begin{eqnarray}
\hskip-1cm  & & \frac{1}{\alpha_R} - \sum_{n=0}^{\infty} \frac{|\phi_n(0)|^2 (E_k + \delta E_k +\mu^2)}{(E_n+\mu^2)(E_n-(E_k+\delta E_k))} - \int_{\Lambda} \frac{|\chi_{\lambda}(0)|^2 (E_k + \delta E_k +\mu^2)}{(\lambda+\mu^2)(\lambda-(E_k + \delta E_k))} \; d \mu(\lambda) \nonumber \\ & & \hspace{2cm} = \frac{1}{\alpha_R} -
    \sum_{n\neq k} \frac{|\phi_n(0)|^2 (E_k + \delta E_k + \mu^2)}{(E_n+\mu^2)(E_n-(E_k+\delta E_k))} + \frac{|\phi_k(0)|^2 (E_k + \delta E_k + \mu^2)}{(E_k+\mu^2)(\delta E_k)} \nonumber \\  & & \hspace{3cm} - \int_{\Lambda} \frac{|\chi_{\lambda}(0)|^2 (E_k + \delta E_k +\mu^2)}{(\lambda+\mu^2)(\lambda-(E_k + \delta E_k))} \; d \mu(\lambda) = 0 \;.
\end{eqnarray}
If we  expand above expressions in the powers of $\delta E_k$ and multiply the equation by $\alpha_R \delta E_k$, we get
\begin{eqnarray} \hskip-1cm
& & \delta E_k + \alpha_R |\phi_k(0)|^2 + \frac{\alpha_R |\phi_k(0)|^2}{E_k+\mu^2} \delta E_k - \alpha_R \delta E_k \sum_{n\neq k} \frac{|\phi_n(0)|^2 (E_k+\mu^2)}{(E_n-E_k)(E_n+\mu^2)} \left( 1+ \frac{\delta E_k}{E_n-E_k} +O(\delta E_{k}^{2})\right) \nonumber  \\ & & - \alpha_R \delta E_k \sum_{n\neq k} \frac{|\phi_n(0)|^2 \delta E_k}{(E_n-E_k)(E_n+\mu^2)} \left( 1+ \frac{\delta E_k}{E_n-E_k} +O(\delta E_{k}^{2})\right)  
\nonumber  \\ & & - \alpha_R \delta E_k  \int_{\Lambda}\frac{|\chi_{\lambda}(0)|^2 (E_k+\mu^2)}{(\lambda-E_k)(\lambda+\mu^2)} \left( 1+ \frac{\delta E_k}{\lambda-E_k} +O(\delta E_{k}^{2})\right) \; d \mu(\lambda) \nonumber  \\ & & - \alpha_R \delta E_k  \int_{\Lambda}\frac{|\chi_{\lambda}(0)|^2 \delta E_k}{(\lambda-E_k)(\lambda+\mu^2)} \left( 1+ \frac{\delta E_k}{\lambda-E_k} +O(\delta E_{k}^{2})\right) \; d \mu(\lambda) + O(\delta E_{k}^3)  = 0 \;. \label{perturbativebsrenormcase}
\end{eqnarray}
Let us assume that $\delta E_k$ can be expandable in the power series of $\alpha_R$, that is, $\delta E_k = E_{k}^{(1)} + E_{k}^{(2)} + \cdots$, where $E_{k}^{(n)}$ corresponds to the change in the bound state energy of order $\alpha_{R}^n$. Then, solving $E_{k}^{(1)}$ and $E_{k}^{(2)}$ term by term, we obtain
\begin{eqnarray} \hskip-1cm
    E_{k}^{(1)} & = &  - \alpha_R |\phi_k(0)|^2 \;, \label{E1perturbationrenorm} \\   \hskip-1cm E_{k}^{(2)} & = &  \alpha_{R}^{2} |\phi_k(0)|^2 \left(\frac{|\phi_k(0)|^2}{E_k+\mu^2} - \sum_{n\neq k} \frac{|\phi_n(0)|^2 (E_k+\mu^2)}{(E_n-E_k)(E_n+\mu^2)} -\int_{\Lambda}\frac{|\chi_{\lambda}(0)|^2 (E_k+\mu^2)}{(\lambda-E_k)(\lambda+\mu^2)} d \mu(\lambda) \right) \label{E2perturbationrenorm} \;.
\end{eqnarray}
It is important to notice that the first order result in the bound state energy is the same as the case where the renormalization is not required, except that $\alpha$ is replaced by the renormalized coupling constant $\alpha_R$. However, the above second order result in the bound state energy is completely different from the case where the renormalization is not required.  
For the wave function, we formally obtain the same formula (\ref{wavefunctionnonrenormcase}) for the wave function expansion in $\delta E_k$. Substituting the results for first and second order eigenvalues given above, we find 
\begin{align}
 \psi_{k}^{(0)}(x) & =  e^{-i \theta_k + i \pi} \phi_{k}(x) \;, \\  \psi_{k}^{(1)}(x) & = \alpha_R \phi_k(0) e^{-i \theta_k + i \pi} \left( \sum_{n \neq k} \frac{\phi_n(x) \overline{\phi_n(0)}}{E_n-E_k}  + \int_{\Lambda} \frac{\chi_{\lambda}(x) \overline{\chi_{\lambda}(0)}}{\lambda- E_k}  d \mu(\lambda) \right) \;, 
 \end{align}
and 
\begin{align} \psi_{k}^{(2)}(x) & =  - \frac{\alpha_R^2}{2} \phi_k(x) e^{-i\theta_k+i \pi} |\phi_k(0)|^2 \left(\sum_{n\neq k} 
    \frac{|\phi_n(0)|^2}{(E_n-E_k)^2} + \int_{\Lambda}  \frac{|\chi_{\lambda}(0)|^2}{(\lambda-E_k)^2} d \mu(\lambda) \right) 
    \nonumber \\ &  \hskip-1cm +  \alpha_R^2 \sum_{n \neq k} \frac{\phi_n(x) \overline{\phi_n(0)}}{(E_n - E_k)} \phi_k(0) e^{-i \theta_k + i \pi} \Bigg[\sum_{m\neq k} \frac{|\phi_m(0)|^2 (E_k+\mu^2)}{(E_m-E_k)(E_m+\mu^2)} + \int_{\Lambda} \frac{|\chi_{\lambda}(0)|^2(E_k+\mu^2)}{(\lambda-E_k)(\lambda+\mu^2)}  d \mu(\lambda)    \nonumber \\ & \hspace{4cm} - \frac{ |\phi_k(0)|^2 (E_n+\mu^2)}{(E_n-E_k)(E_k+\mu^2)} \Bigg]  \nonumber \\ &  \hskip-1cm
   +  \alpha_R^2  \int_{\Lambda} \frac{\chi_{\lambda}(x) \overline{\chi_{\lambda}(0)}}{(\lambda-E_k)}   \phi_k(0) e^{-i \theta_k + i \pi} \Bigg[ \sum_{m\neq k} \frac{|\phi_m(0)|^2}{E_m-E_k} + \int_{\Lambda} \frac{|\chi_{\lambda'}(0)|^2}{\lambda'-E_k}  d \mu(\lambda') \nonumber \\ & \hspace{4cm} - \frac{ |\phi_k(0)|^2 (\lambda+\mu^2)}{(\lambda-E_k)(E_k+\mu^2)} \Bigg]   d \mu(\lambda) \;.
\end{align}

\begin{myremark}
We note that the perturbation theory results for the renormalized problem can be recovered from an approach
inspired by field theory. For simplicity, we discuss a Hamiltonian $H_0$ with purely discrete spectrum, that could be for example a free particle moving on a manifold or a harmonic oscillator perturbed by a delta interaction. All the other cases discussed in this paper can similarly be shown. 

We introduce a cut-off in the eigenvalues of $H_0$,   $N$ and define $\delta_N(x,a)$ as our regularized delta-interaction on a manifold.
Here 
\begin{eqnarray}
\delta_N(x,a)=\sum_{n} e^{-n/N}\overline{\phi_n(x)}\phi_n(a)    \;.
\end{eqnarray}
Clearly as $N\to \infty$ we get $\delta_N(x,a)\to \delta(x,a)$ (in the weak sense). Let us define 
\begin{eqnarray}
  {1\over \alpha}= {1\over \alpha_R}+ G_0(a,a|-\mu^2;N) \;,  
\end{eqnarray}
where 
\begin{eqnarray}
G_0(a,a|-\mu^2;N)=\sum_{n} {e^{-n/N}|\phi_n(a)|^2\over E_n+\mu^2} \;,    
\end{eqnarray}
or equivalently,
\begin{eqnarray}
   \alpha= {\alpha_R\over 1+\alpha_RG_0(a,a|-\mu^2;N)} \;.
\end{eqnarray}
Now, within the usual perturbative approach, we assume  $\alpha_R$
 is a {\it  formal parameter, that can be made arbitrarily small so as to organize our expansions accordingly}.
This means we should order everything according to powers of $\alpha_R$ and formally expand, which gives 
\begin{eqnarray}
\alpha=\alpha_R-\alpha_R^2G_0(a,a|-\mu^2;N) + O(\alpha_R^3) \;.    
\end{eqnarray}
The formal  $\delta$-interaction term in the Hamiltonian now becomes,
\begin{eqnarray}
-\left[\alpha_R-\alpha_R^2G_0(a,a|-\mu^2;N) + O(\alpha_R^3)\right]\delta_N(x,a) \;.    
\end{eqnarray}
Hence, a standard perturbative expansion, organized according to the powers of $\alpha_R$ will have mixed terms as {\it the interaction term now is a power series} in the  small parameter $\alpha_R$. (there is a sense in which we can imagine $\alpha_R$ as much smaller than any other quantity and then use some analytic continuation arguments to take the limit $N\to \infty$ while keeping the formulae intact).
First order perturbation now keeps only the first term of  the interaction,
\begin{eqnarray}
   E^{(1)}_{N;k}=-\alpha_R\int_{\mathbb{R}^d} d^d x \; \overline{\phi_k(x)} \delta_N(x,a) \phi_k(x) \;, 
\end{eqnarray}
where $d=2$ or $d=3$. 
The limit $N\to \infty$ reproduces the answer (\ref{E1perturbationrenorm}).
The second order term has two parts, one is from the interaction that should be treated at first order as it has $\alpha_R^2$ and then the leading term that should be treated at second order:
\begin{eqnarray}
     E_k^{(2)} &=&   \alpha_R^2 \sum_{m\neq k} {\int_{\mathbb{R}^d} d^d x \overline{\phi_k(x)} \delta_N(x,a)\phi_m(x) \int_{\mathbb{R}^d} d^d x' \overline{\phi_m(x')} \delta_N(x',a)\phi_k(x')\over E_k-E_m} \nonumber\\ 
& & \hspace{2cm}  + \alpha_R^2 \sum_{m}  {e^{-m/N} |\phi_m(a)|^2\over E_m+\mu^2} \int_{\mathbb{R}^d} d^d x \; \overline{\phi_k(x)} \delta_N(x,a) \phi_k(x) \;.
\end{eqnarray} 
In the limit $N\to \infty$ leaving aside the convergence issues, we see that one gets,
\begin{eqnarray}
  E_k^{(2)} &=&  \alpha_R^2 \sum_{m\neq k} {\overline{\phi_k(a)} \phi_m(a)\overline{\phi_m(a)} \phi_k(a)\over E_k-E_m}
 + \alpha_R^2 \sum_{m}  { |\phi_m(a)|^2\over E_m+\mu^2} \overline{\phi_k(a)}  \phi_k(a) \nonumber\\
&\ & = \alpha_R^2 \Big(|\phi_k(a)|^2 (E_k+\mu^2) \sum_{m\neq k} {|\phi_m(a)|^2\over (E_k-E_m)(E_m+\mu^2)}+{|\phi_k(a)|^2 |\phi_k(a)|^2\over E_k+\mu^2} \Big),\nonumber
\end{eqnarray} 
 where the last term comes from isolating the $m=k$ term from the $G_N(a,a|-\mu^2)$ part.
So our direct approach essentially provides a sound basis for this regularized perturbation theory cancellations and expansions.
\end{myremark}

\section{Possible Generalizations}
\label{PossibleGeneralizations}
\subsection{$N$ Center Case}
\label{NCenterCase}

It is possible to generalize our approach to $N$-pointlike $\delta$ interaction case. When there is no need for renormalization, we will assume that all the couplings are actually positive (hence corresponds to the attractive case). Let us enumerate these points  as $a_1, a_2,..., a_N$ with $a_i \neq a_j$ whenever $i \neq j$ and the associated couplings with $\alpha_1,\alpha_2,...\alpha_N$. We will proceed recursively, and suppose that $H_0$ has the same spectral properties discussed in Section \ref{section2}. If some set of eigenfunctions satisfies $\phi_k(a_1)\neq 0$ then they lead to a shift of this eigenvalue to a new value $E_k^{*_1}$ in between $E_{k-1}$ and $E_k$. For these, the eigenfunctions change to ${\mathcal N}_0G_0(x,a_1|E_k^{*_1})$
where ${\mathcal N_0}$ is the normalization constant. So, we have a new Hamiltonian $H_1$ with a new set of discrete states $\phi_k^{(1)}$ (and possibly a new set of continuum states that we do not venture into calculating). We now are back to the initial case, if we add the $\delta$ interaction at $a_2$, the same construction is repeated, and new eigenvalues $E_k^{*_2}$ for $\phi_k^{(1)}(a_2)\neq 0$ fall in between $E_{k-1}^{*_1}$ and $E_k^{*_1}$. The poles associated with the old eigenvalues are removed, and there are new wave functions given by ${\mathcal N}_1G_1(x, a_2|E_k^{*_2})$. We can now proceed recursively, and define $H_{N}=H_{N-1}-\alpha_N \delta(x,a_N)$ and thus find the Green's function \cite{Besprosvany, recursivegreen}
\begin{equation}
    G_N(x,y|E)=G_{N-1}(x, y|E)+G_{N-1}(x, a_N|E)\Phi^{-1}(E)G_{N-1}(a_N,y|E),
\end{equation}
with $\Phi(E)={1\over \alpha_N}- G_{N-1} (a_N,a_N|E)$. By this recursive argument, the new eigenvalues $E_k^{*_N}$ are in between the eigenvalues $E_{k-1}^{*_{N-1}}$ and $E_k^{*_{N-1}}$. The resulting eigenfunction 
$\phi_k^{(N)}(x)$ is given by ${\mathcal N}_{N-1}G_{N-1}(x, a_N| E_k^{*_N})$, note that 
$G_{N-1}(x,a_N|E_k^{*_N})=G_0(x, a_N| E_k^{*_N})+\sum_{i,j=1}^{N-1} G_0(x, a_i|E_k^{*_N})\Phi^{-1}_{ij}(E_k^{*_N})G_0(a_j,a_N|E_k^{*_N})$,
where $\Phi_{ij} (E)=\frac{1}{\alpha_i} \delta_{ij}- G_0(a_i,a_j|E)$ is the matrix formed by the point centers $a_1,...a_{N-1}$.

Incidentally, the above derivation {\it does not make use of the finiteness of the matrix $\Phi$}, even in case the matrix $\Phi(E)$ requires a renormalization, our derivation remains valid.

\subsection{$\delta$ Interaction Supported on Curves in Plane} 
\label{DeltaPotentialSupportedonCurvesinPlane}

Note that none of the derivations actually rely on the interaction being concentrated at a point, we can generalize to the curve case easily. 
In fact, by extending the above discussion, we can accommodate multiple non-intersecting curves and points cases easily. To make the discussion simpler we consider a single rectifiable curve $\Gamma$ in the plane first, and assume that the spectrum of $H_0$ consists of only a discrete part. 
There are various ways to define the above Hamiltonian in a mathematically rigorous way. One  way is to interpret the above formal interaction by the quadratic form 
\begin{eqnarray}
    \int_{\mathbb{R}^2} |\nabla \psi|^2 d^2 x - \alpha \int_{\Gamma} |\psi|^2 d s \;,
\end{eqnarray}
in the case where $H_0$ is the Laplacian. Then, one can prove that there is a self-adjoint Hamiltonian associated with this quadratic form \cite{brasche1994schrodinger, ExnerYoshitomi2002, Kondej2016}. The other way is to impose the continuity and the jump discontinuity conditions of the normal derivatives at $\Gamma$, (see Remark 4.1 in \cite{brasche1994schrodinger} and \cite{exner2003spectra}). The other methods are based on using scaled potentials \cite{exner2001geometrically}, or direct construction of the resolvent operator \cite{brasche1994schrodinger, posilicano2001krein}. The physical motivation for studying such Hamiltonians is to give a realistic model for trapped electrons due to interfaces between two different semiconductor materials, which are known as leaky graphs, curves or surfaces in the literature \cite{Exner}.

In this section, we  consider rank one modification (perturbations) of $H_0$ in the sense described in \cite{AlbeverioKurasov} and the Hamiltonian is formally given by 
\begin{eqnarray}
H_0 - \alpha |\Gamma \rangle \langle \Gamma | \;, \label{formalHamiltoniancirclepoint}
\end{eqnarray}
where $H_0 =-\frac{\hbar^2}{2m}\Delta + V$ and we have used the Dirac notation for the inner products, and the Dirac delta function $\delta_{\Gamma}$ supported by the curve $\Gamma$ with length $L$ is defined by their action on test functions $\psi$ \cite{Appel}
\begin{eqnarray}
\langle \delta_{\Gamma}| \psi \rangle = \langle \Gamma | \psi \rangle & := & \frac{1}{L(\Gamma)} \int_{\Gamma} \psi \; d s \;, 
\end{eqnarray}
where $ds$ is the integration element over the curve $\Gamma$ and $|\Gamma \rangle \langle \Gamma|$ written in Dirac's bra-ket notation is $\langle \delta_{\Gamma}, \cdot \rangle \delta_{\Gamma}$.

It is well known that the resolvent of free Hamiltonians modified by $\delta$ interaction supported on a curve can be expressed by some explicit formulae involving the resolvent of the free Hamiltonian, and they are known as Krein's formula in the literature \cite{Albeverio2012solvable, AlbeverioKurasov, Exner}. Instead of free Hamiltonian we have here a general Schr\"{o}dinger operator but the formula for the resolvent would be exactly the same as before. Hence, we obtain 
\begin{eqnarray}
    R(E)=R_0(E)+ \frac{1}{\Phi(E)} R_{0}(E) |\Gamma \rangle \langle \Gamma | R_0(E) \;,
\end{eqnarray}
where $\Phi(E)=\frac{1}{\alpha}-G_0(\Gamma, \Gamma|E)$ and $G_0(\Gamma, \Gamma|E)=\frac{1}{L^2} \iint_{\Gamma \times \Gamma} G_0(\gamma(s),\gamma(s')) ds \; ds'$. 
This can be expressed in terms of the Green's functions as
\begin{eqnarray}
    G(x,y|E)=G_0(x,y|E) + \frac{1}{\Phi(E)} G_{0}(x, \Gamma|E) G_0(\Gamma, y|E) \;.
\end{eqnarray}
Here $G_0(x, \Gamma|E)=\frac{1}{L} \int_{\Gamma} G_0(x,\gamma(s)) ds$. A generalization of such $\delta$ interactions supported on curves embedded in manifolds is studied in \cite{KaynakTurgut1}.

Let us define
$g(x,y|E)=\sum_{n\neq k} \phi_n(x)\phi_n(y)(E-E_n)^{-1}$ as well as projections onto  the curve as $\langle x|g(E)| \Gamma\rangle=\sum_{n\neq k}  \phi_n(x)\langle \Gamma| \phi_k\rangle (E-E_n)^{-1}$. For clarity we only assume a discrete spectrum, since generalizing to continuous spectrum is not difficult.
These are holomorphic functions of $E$ {\it in a sufficiently small neighborhood} of $E_k$. By following the same steps as we have done before, and using the expansion of the Green's function we find
\begin{eqnarray}
    G(x,y|E)&=&g(x,y|E)+\frac{\phi_k(x)\phi_k(y)}{E-E_k} + \frac{{\phi_k(x)\langle \phi_k|\Gamma\rangle\over E-E_k}{\langle \Gamma|\phi_k\rangle \phi_k(y)\over E-E_k}}{{1\over \alpha}-\langle \Gamma| g(E)|\Gamma\rangle-|\langle \Gamma|\phi_k\rangle|^2(E-E_k)^{-1}}\nonumber\\
    &\ &     +\frac{\langle x|g(E)|\Gamma\rangle{\langle \Gamma|\phi_k\rangle \phi_k(y)\over E-E_k}}{{1\over \alpha}-\langle \Gamma| g(E)|\Gamma\rangle-|\langle \Gamma|\phi_k\rangle|^2(E-E_k)^{-1}}
    +\frac{{\phi_k(x)\langle \phi_k|\Gamma\rangle\over E-E_k}\langle\Gamma|g(E)|y\rangle}{{1\over \alpha}-\langle \Gamma| g(E)|\Gamma\rangle-|\langle \Gamma|\phi_k\rangle|^2(E-E_k)^{-1}}\nonumber\\
    &\ & + \frac{\langle x|g(E)|\Gamma\rangle\langle \Gamma|g(E)| y\rangle}{{1\over \alpha}-\langle \Gamma| g(E)|\Gamma\rangle-|\langle \Gamma|\phi_k\rangle|^2(E-E_k)^{-1}} \;.
\end{eqnarray}
Here we have assumed that there is  a continuous representative for each eigenvector, so that the expression $\langle \Gamma|\phi_k\rangle= \frac{1}{L}\int_{\Gamma} \phi_k(\gamma(s)) ds$ is well defined, this is true even for the generalized eigenvectors of  a Laplacian modified by some potential which satisfies some conditions \cite{BeliyKovalenkoSemenov} (indeed for  the free Hamiltonian, the domain of $H_0$ is the Sobolev space and these expressions for arbitrary vectors in the domain are always well-defined by the Sobolev embedding theorem \cite{HaroskeTriebel}). 
Just as before, multiplying by $(E-E_k)$ the numerators and denominators, and expanding around $E_k$, we see the cancellation of the pole at $E_k$, as long as $\langle \Gamma|\phi_k\rangle\neq 0$.

In this case we have the new wave functions
given by 
\begin{equation}
    \psi_k(x)={\mathcal N} G_0(x, \Gamma|E_k^*)={\mathcal N} \int_{\Gamma}  G_0(x, \gamma(s)|E_k^*) d s \;. 
\end{equation}
Eigenvalues and eigenvectors, for multiple curves and points are now  constructed by iterating the above set of arguments for each additional rank one perturbation.

\begin{myremark}
In all our discussions,  it is {\it essential} to have elliptic operators with some special properties (typically a summability condition on the potentials) to allow for generalized eigenvectors with certain regularity properties, as we evaluate them at a point or integrate them over a curve.
\end{myremark}

\subsection{A Particle  in a Compact Manifold under the Influence of a   $\delta$ Interaction }
\label{AParticleMovinginaCompact ManifoldandInteractingwithdeltaPotential}

Another possibility for $H_0$ could correspond to a system where a  particle is intrinsically moving on a two or three dimensional compact and connected manifold $M$ with the metric structure $g$:
\begin{eqnarray}
    (H_0 \psi)(x)= - \frac{\hbar^2}{2m} \left( \frac{1}{\sqrt{\det g}} \sum_{i,j=1}^{d} \frac{\partial}{\partial x^i} \left( \sqrt{\det g} g^{ij} \frac{\partial \psi(x)}{\partial x^j}\right)  \right)  \;,
\end{eqnarray}
where $x= (x^1, \ldots, x^d)$ are the local coordinates, and $g^{ij}$ are the components of inverse of the metric $g$. Then it is well known 
\cite{Rosenberg, chavel2} that there exists a complete orthonormal system of $C^\infty$
eigenfunctions $\{\phi_n \}_{n=0}^{\infty}$ in $L^2(M)$ and
the spectrum $\sigma(H_0)=\{E_n\} = \{0 = E_0
\leq E_1 \leq E_2 \leq \dots\}$, with $E_n$ tending
to infinity as $n \rightarrow \infty$ and each eigenvalue has
finite multiplicity: $H_0 \phi_n=E_n \phi_n$.
Some eigenvalues are repeated according to their multiplicity. The
multiplicity of the first eigenvalue $E_0=0$ is one and the
corresponding eigenfunction is constant.

When we add a point like $\delta$ interaction on such Schr\"{o}dinger operator $H_0$, physically we model  a particle moving intrinsically in a two dimensional compact manifold $\mathcal{M}$ (without boundary) and interacting with a $\delta$ source located at some point $a$ in $\mathcal{M}$. This problem requires renormalization due to the short distance singular behavior of the free Green's function on a manifold, which can be similarly seen from the short time expansion of the heat kernel $K_t(a,a) \sim 1/t$ as $t \to 0^+$ and the formula (\ref{intrepgreen}) for Riemannian manifolds \cite{Grigoryan}. This is the same singular structure of the diagonal free Green's function and the reason for this is based on the fact that the divergence here appears due to the short time behavior of the heat kernel and the manifold $\mathcal{M}$ locally looks like a flat space around the point where the $\delta$ interaction is located. This has been discussed in our previous works \cite{AltunkaynakErmanTurgut, ErmanTurgut, erman2017number, CaglarErmanTurgut} for finitely many $\delta$ centers. Then, the Green's function can be similarly expressed as
\begin{eqnarray}
G(x,y|E)= G_0(x,y|E) + \frac{G_0(x,a|E) G_0(a,y|E)}{\frac{1}{\alpha_R}- \sum_{n=0}^{\infty} \frac{|\phi_n(a)|^2(E+\mu^2)}{(E_n-E)(E_n+\mu^2)}} \;,
\end{eqnarray}
for any points $x,y$ and $a$ in $\mathcal{M}$. Note that this expression is formally the same as our previous formula for a two dimensional problem without the need for a continuous spectrum. Then, all the results that we have obtained for previous problems follow easily in this case as well. Consequently, we do not repeat our arguments, instead in the next section we concentrate on a slightly different model on a compact manifold, and there we provide some more details. The convergence of the series $\sum_{n=0}^{\infty} \frac{|\phi_n(a)|^2(E+\mu^2)}{(E_n-E)(E_n+\mu^2)}$ will be shown in Appendix A.

\subsection{A Semirelativistic Model on a Two Dimensional Compact Manifold}

In our previous work \cite{caglarturgut},  a possible relativistic model of $\delta$ interactions on a two dimensional compact Riemannian manifold $(\mathcal{M},g)$ is proposed, in a second quantized language the model Hamiltonian is written as
\begin{equation}
    H= \frac{1}{2}\int_{\mathcal{M}} d_g^2x : \phi^\dagger (-\nabla^2_g+m^2)\phi+\pi^2: \; - \, \alpha \, \phi^{(-)}(a)\phi^{(+)}(a) \; ,
\end{equation}
where $d_g^{2}x = \sqrt{\det g} \; dx_1 \, dx_2$ and $(x_1,x_2)$ is the local coordinates. The positive and the negative frequency part of the bosonic field $\phi$ is denoted by $\phi^{(-)}$ and $\phi^{(+)}$, respectively. They can be expanded in terms of the creation and annihilation operators indexed by $n$, that is,
\begin{eqnarray}
    \phi^{(+)}(x) & = & \sum_{n=0}^{\infty} {f_n(x)\over \sqrt{2\omega_n}}a_n , \\ 
       \phi^{(-)}(x) & = & \sum_{n=0}^{\infty} {\overline{f_{n}(x)}\over \sqrt{2\omega_n}}a^{\dagger}_n
\end{eqnarray}
where $f_n(x)$ are the eigenfunctions of the Laplace Beltrami operator $\nabla_{g}^2= \frac{1}{\sqrt{\det g}} \sum_{i,j=1}^{2} 
\frac{\partial}{\partial x^i} \left( g^{ij} \sqrt{\det g} \frac{\partial}{\partial x^j}\right)$:
\begin{equation}
    -\nabla_{g}^2 f_{n}(x) = e_n^2 f_{n}(x) \;,
\end{equation}
and $\omega_{n}= \sqrt{e_n^2+m^2}$, note that we use units where $\hbar=1$ and $c=1$. Therefore $\omega_n$ and energy has the same units.  We remark that the index $n$ plays the same role as momentum  $k$ in flat space,  although there is {\it no relation} with a (possible) translation symmetry in general. If we restrict our model to a one particle sector 
\begin{equation}
    |\psi\rangle= \sum_{n=0}^{\infty} \psi_n {a^\dagger_n\over \sqrt{2 \omega_n}}|0\rangle,
\end{equation}
and following the same steps as in \cite{caglarturgut} we will find the Green's function restricted to one-particle sector. To find the Green's function, we need to solve  in the eigenfunction basis,  
\begin{equation}
(H-E)|\psi\rangle = |\chi\rangle \;,
\end{equation}
for a given state $|\chi\rangle$. As a result, the one-particle sector Green's function can be found by solving, 
\begin{equation}
\left(\sqrt{e_{n}^2 + m^2 }-E \right)\psi_{n}- \alpha \sum_{l=0}^{\infty} \frac{\overline{f_{n}(a)}}{\sqrt{\omega_{l}}} \frac{f_{l}(a)}{\sqrt{\omega_{l}}}\psi_{l} = \chi_{n}
\;.\end{equation}
The solution $\psi_n$ of this equation is given by 
\begin{equation}
\psi_{n} = \frac{1}{(\omega_{n} -E)}\chi_{n} + \frac{1}{(\omega_{n}-E)} \overline{f_{n}(a)} \bigg[\frac{1}{\alpha}- \sum_{l=0}^{\infty} \frac{|f_{l}(a)|^2}{\omega_{l}(\omega_{l}-E)}\bigg]^{-1}\sum_{r=0}^{\infty} \frac{f_{r}(a) }{\omega_{r}}\frac{1}{(\omega_{r}-E)}\chi_{r} \;.
\end{equation}
To read the Green's function, notice that
\begin{eqnarray} & & 
\sum_{n=0}^{\infty} f_{n}(x) \frac{\psi_{n}}{\sqrt{\omega_n}} = \psi(x)  =   \sum_{n=0}^{\infty} \int_{\mathcal{M}} d_{g}^{2} y \;  \frac{f_{n}(x) \overline{f_{n}(y)}}{(\omega_n-E)} \chi(y) \nonumber\\
 & & \hspace{1cm} +  \sum_{n=0}^{\infty} \frac{f_{n}(x)}{ (\omega_{n}-E)} \frac{\overline{f_{n}(a)}}{\sqrt{\omega_{n}}} \bigg[\frac{1}{\alpha}- \sum_{l} \frac{|f_{l}(a)|^2}{\omega_{l}(\omega_{l}-E)}\bigg]^{-1}  \sum_{r} \int_{\mathcal{M}} d_{g}^{2} y \; \frac{f_{r}(a) }{\sqrt{\omega_{r}}}\frac{\overline{f_{r}(y)}}{(\omega_{r}-E)} \chi(y) \;, \end{eqnarray}
where we have used $\chi_n=\sqrt{\omega_n} \int_{\mathcal{M}} d_{g}^{2} x \; \chi(x) \overline{f_{n}(x)}$ and orthogonality of eigenfunctions $f_n$. We can then write this as an equation,
\begin{eqnarray}
  \psi(x)=  \int_{\mathcal{M}} d_{g}^{2} y \; G_0(x,y|E)\chi(y) +\tilde{G}_0(x,a|E) \Phi(E)^{-1}  \int dy \; \tilde{ G}_0(a,y|E)\chi(y) \;,
\end{eqnarray}
where
\begin{eqnarray}
    G_0(x,y|E) & = & \sum_{n=0}^{\infty} \frac{f_{n}(x) \overline{f_{n}(y)}}{(\omega_{n}-E)}  \\
    \tilde{ G}_0(x,y|E) & = & \sum_{n=0}^{\infty} \frac{f_{n}(x)}{ (\omega_{n}-E)} \frac{\overline{f_{n}(y)}}{\sqrt{\omega_{n}}} \\ 
  \Phi(E) & = & \frac{1}{\alpha}- \sum_{n=0}^{\infty} \frac{|f_{n}(a)|^2}{\omega_{n}(\omega_{n}-E)} \;.
\end{eqnarray}
Then we can find the Green's function as
\begin{equation}
G(x,y|E)= G_0(x,y) + \tilde{G}_0(x,a|E)\Phi^{-1}(E) \tilde{G}_0(a,y|E)
.\end{equation}
Note that to demonstrate the 
cancellation of original poles, we need to carefully look at the analytic structure of $G(x,y)$.
\begin{equation}
G(x,y)=\sum_{n=0}^{\infty}  \frac{f_{n}(x) \overline{f_{n}(y)}}{(\omega_{n}-E)} + \sum_{n=0}^{\infty} \frac{f_{n}(x)}{ (\omega_{n}-E)} \frac{\overline{f_{n}(a)}}{\sqrt{\omega_{n}}}\Phi^{-1}(E) \sum_{r=0}^{\infty}  \frac{f_{r}(a) }{\sqrt{\omega_{r}}}\frac{ \overline{f_{r}(y)}}{(\omega_{r}-E)}
\end{equation}
Indeed, the sum in the definition of $\Phi$ is divergent. This can be seen by following the same line of arguments discussed in \cite{caglarturgut}. We will summarize it for the sake of completeness. First we rewrite the expression $\frac{1}{\omega_n (\omega_n-E)}$ as
\begin{eqnarray}
    \frac{1}{\omega_n (\omega_n-E)} = -\frac{1}{E} \left( \frac{1}{\omega_n}- \frac{1}{\omega_n -E}\right) =- \frac{1}{E} \int_{0}^{\infty} e^{-s \omega_n}(1-e^{s E}) \; ds \;,
\end{eqnarray}
where we assume that $\Real{(\omega_n-E)}>0$. The later formulae defining the resolvent in  terms of the heat kernel   is often  assumed to be analytically continued onto the complex plane. Using the subordination identity 
\begin{eqnarray}
    e^{-s A} = \frac{s}{2\sqrt{\pi}} \int_{0}^{\infty} e^{-s^2/(4u)-u A^2} \; \frac{d u}{u^{3/2}} \;, \label{subordination}
\end{eqnarray}
and the eigenfunction expansion of the heat kernel associated with the Laplace-Beltrami operator on Riemannian manifolds \cite{Rosenberg}
\begin{eqnarray} \label{eigenfunctionexpansionofheatkernel}
    K_u(x,y)= \sum_{n=0}^{\infty} e^{-u e_n^2} f_{n}(x) \overline{f_{n}(y)} \;,
\end{eqnarray}
we obtain
\begin{eqnarray} & & 
    \sum_{n=0}^{\infty} \frac{|f_{n}(a)|^2}{\omega_{n}(\omega_{n}-E)}  = -\frac{1}{E} \int_{0}^{\infty} \bigg( \int_{0}^{\infty}  K_u(a,a) e^{-u m^2 + s^2/(4u)} \; \frac{d u}{u^{3/2}} \bigg) \frac{(1-e^{s E}) s}{2\sqrt{\pi}}  \; d s \;.
\end{eqnarray}
Using integration by parts and scaling the integration variable, we find
\begin{eqnarray} & & 
    \sum_{n=0}^{\infty} \frac{|f_{n}(a)|^2}{\omega_{n}(\omega_{n}-E)}  = \frac{1}{\sqrt{\pi}} \int_{0}^{\infty} \bigg( \int_{0}^{\infty}  K_u(a,a) e^{-u m^2 + s E \sqrt{u}} \; d u \bigg) e^{-s^2/4} \; d s \;.
\end{eqnarray}
Due to the short ``time" asymptotic behavior of the diagonal heat kernel $K_u(a,a) \sim \frac{1}{u}$ for two dimensional Riemannian manifolds, the integral over $u$ is divergent. For this reason, we apply the idea of renormalization and introduce a short ``time" cut off $\epsilon$ in the lower limit of $u$ integral. This corresponds to the sum over $\tau$ up to a cut off $N$. Then we choose the coupling constant depending on this cut off in such a way that 
\begin{eqnarray}
    \frac{1}{\alpha(N)} = \frac{1}{\alpha_R} + \sum_{n=0}^{N}  \frac{|f_{n}(a)|^2}{\omega_{n}(\omega_{n}+\mu^2)} \;,
\end{eqnarray}
where $-\mu^2$ can be  related to the experimentally measured bound state energy of the particle below the lowest eigenvalue $\omega_0$ (in the presence of single Dirac $\delta$ center) as we will see. To this end we remove the cut off by sending $N\to \infty$ in $\Phi$ and get the renormalized $\Phi$
\begin{eqnarray}
\Phi_R &=&  \frac{1}{\alpha_R}+ \sum_{n=0}^{\infty} \bigg[ \frac{|f_{n}(a)|^2}{\omega_{n}(\omega_{n}+\mu^2)} -\frac{|f_{n}(a)|^2}{\omega_{n}(\omega_{n}-E)} \bigg]\nonumber\\
&=&  \frac{1}{\alpha_R}+\sum_{n=0}^{\infty} |f_{n}(a)|^2 \frac{(-E-\mu^2)}{\omega_{n} ( \omega_{n} +\mu^2)( \omega_{n} -E)}\nonumber\\
&=&   \frac1{\alpha_R}- ( E+\mu^2)\sum_{n=0}^{\infty} \frac{ |f_{n}(a)|^2}{\omega_{n} ( \omega_{n} +\mu^2)( \omega_{n} -E)} \;.
\end{eqnarray}
This expression when the real part of $\omega_0-E$ is positive can be shown to be finite by using heat kernel estimates (which is done in \cite{caglarturgut}), one can furthermore show than as long as $E\neq \omega_n$ this expression is well-defined and analytic.
We take $E$ around $E_k$ and isolate this particular pole in our expressions,
\begin{eqnarray} & & \sum_{n \ne k }  \frac{f_{n}(x) \overline{f_{n}(y)}}{(\omega_{n}-E)}  +  \frac{f_{k}(x) \overline{f_{k}(y)}}{(\omega_{k}-E)}  + \bigg(\sum_{n \ne k} \frac{f_{n}(x)}{ (\omega_{n}-E)} \frac{\overline{f_{n}(a)}}{\sqrt{\omega_{n}}} + \frac{f_{k}(x)}{ (\omega_{k}-E)} \frac{\overline{f_{k}(a)}}{\sqrt{\omega_{k}}}\bigg)\nonumber\\
& & \hspace{1cm} \times \frac{1}{  \frac1{\alpha_R}- ( E+\mu^2)\Big[\sum_{r \ne  k} \frac{ |f_{r}(a)|^2}{\omega_{r} ( \omega_{r} +\mu^2)( \omega_{r} -E)}  +  \frac{ |f_{k}(a)|^2}{\omega_{k} ( \omega_{k} +\mu^2)( \omega_{k} -E)}\Big]} \nonumber \\ & & \hspace{2cm} \times \bigg(\sum_{l \ne k}  \frac{f_{l}(a) }{\sqrt{\omega_{l}}}\frac{\overline{f_{l}(y)}}{(\omega_{l}-E)} +  \frac{f_{k}(a) }{\sqrt{\omega_{k}}}\frac{\overline{f_{k}(y)}}{(\omega_{k}-E)}\bigg) \;. 
\end{eqnarray}
Note that here the singular part comes from,
\begin{equation}
\frac{f_{k}(x) \overline{f_{k}(y)}}{(\omega_{k}-E)}- \frac{|f_{k}(a)|^2 f_{k}(x)\overline{f_{k}(y)}}{\omega_{k}(\omega_{k}-E)^2 \big[ (E+\mu^2)\frac{|f_{k}(a)|^2}{\omega_{k}(\omega_{k}-E)(\omega_k +\mu^2)}+ \text{regular terms} \big ]}+ \text{regular terms}
\;.\end{equation}
We now use $E+\mu^2=E-\omega_k+\omega_k+\mu^2$ in the numerator and reduce this to 
\begin{equation}
\frac{f_{k}(x) \overline{f_{k}(y)}}{(\omega_{k}-E)}- \frac{|f_{k}(a)|^2 f_{k}(x)\overline{f_{k}(y)}}{\omega_{k}(\omega_{k}-E)^2 \big(\frac{|f_{k}(a)|^2}{\omega_{k}(\omega_{k}-E)}+ \text{regular terms} \big)}+ \text{regular terms}
\;. \end{equation}
Following our previous arguments, it is seen that the pole at $\omega_k$ is now cancelled.
As a result, we see that the interaction we propose has the same property as in the Sturm-Liouville interlacing theorem, the picture we have for the two dimensional non-relativistic problem applies equally well here. In Appendix C, we prove that there is always a bound state below $\omega_0$ for all choices of $\alpha_R$, albeit one should be cautious not to push the bound state below $-m$ for internal consistency of the model.

\section*{Final Remarks} 

In this work, we have studied the spectral properties of the Schr\"{o}dinger operators $H_0$ with $\delta$ interactions, where $H_0$ satisfies some mild conditions, which are usually assumed to hold in most of the quantum systems. In contrast to the typical works in the literature, we  work out the pole structure of $G$ and explicitly show that poles of the initial Green function $G_0$ are removed (as long as the original eigenfunction does not vanish at the location of the $\delta$ function) from $G$, through the use of eigenfunction expansions. We show that new bound state energies can be found by solving equation $\Phi(E)=0$, where $\Phi$ is explicitly defined in terms of the diagonal Green's function if the center of the $\delta$ function is not located at one of the nodes of the initial eigenfunction $\phi_k$. These results are established in one dimension for point like $\delta$ interactions and the reflectionless potential for $H_0$ is studied as an example. A different kind of perturbative approach in finding eigenvalues and eigenfunctions of the full system is presented up to second order in a heuristic way. 
These ideas are then extended to the case, where the renormalization procedure is needed. In that case, we  obtained similar results, except that the perturbative calculations differ at the second order of the coupling constant of $\delta$ interaction. Finally, we summarize some possible further extensions of our results to the more general $\delta$ interactions (supported on curve) and multi-center case, and to the case in which a particle is moving on a compact two dimensional manifold, and to a semi-relativistic model on a compact manifold.

\section*{Data availability statement} 
No new data were created or analysed in this study.

\section*{Acknowledgments}

F. Erman would like to thank H. Uncu for useful discussions on the spectrum of harmonic oscillator modified by $\delta$ interaction, and A. Bohm and M. Gadella for giving helpful insights on the Rigged Hilbert spaces and eigenfunction expansions. O. T. Turgut would like to thank A. Michelangeli  for various discussions on singular interactions and J. Hoppe for discussions and encouragement on this project.  Moreover, O. T. Turgut thanks M. Deserno for the kind invitation to Carnegie-Mellon University where this work has come to completion.

\section*{Appendix A: Convergence of the series $\sum_{n=0}^{\infty} \frac{|\phi_n(a)|^2(E+\mu^2)}{(E_n-E)(E_n+\mu^2)}$}

We verify that Green's function is finite apart from the obvious poles we claim, that is, we would like to prove that 
\begin{equation}
\lim_{N\to \infty} \sum_{n=0}^N \frac{|\phi_n(a)|^2}{(E_n-E)(E_n + \mu^2 )} < \infty \;.
\end{equation}
On a compact manifold, the  Laplace-Beltrami operator $-\nabla_{g}^2$ has a discrete spectrum  $ 0 \le E_0 \le E_1 \le E_2 ... \le E_n  \rightarrow \infty $, and with finite multiplicity. For simplicity, let us first consider the case where $E \in \mathbb{R}^+$. We then  choose a {\it finite} $M$ sufficiently large that $ E_M > 2 \left(E + \frac{\mu^2 }{2} \right)$ and note that $E_n  \geq E_M$ for $n>M$. Then, it is easy to see that for $n \geq M$, (with obvious requirement $N>M$ as $N\to \infty$ eventually), we have
\begin{eqnarray}
\frac{|\phi_n(a)|^2 }{(E_n-E)(E_n + \mu^2 )} <  2 \frac{|\phi_n(a)|^2}{(E_n + \mu^2 )^2} \;. 
\end{eqnarray}
This implies that 
\begin{eqnarray}
\sum_{n \ge M}^N \frac{|\phi_n(a)|^2 }{(E_n-E)(E_n + \mu^2 )} <  2 \sum_{n=0 }^N \frac{|\phi_n(a)|^2}{(E_n + \mu^2 )^2}  \;.
\end{eqnarray}
The right hand side can be expressed as 
\begin{eqnarray}
    \sum_{n=0 }^N \frac{|\phi_n(a)|^2}{(E_n + \mu^2 )^2} = \sum_{n=0}^N  \int_0^\infty  dt \; t \; |\phi_n(a)|^2e^{-(E_n + \mu^2 )t} \;,
\end{eqnarray}
and using the eigenfunction expansion of the heat kernel (\ref{eigenfunctionexpansionofheatkernel}) we obtain in the limit $N \to \infty$
\begin{eqnarray}
       \sum_{n=0}^{\infty} \frac{|\phi_n(a)|^2}{(E_n + \mu^2 )^2}  = \int_0^\infty  dt \; t \; K_t(a,a)e^{- \mu^2 t} \;.
\end{eqnarray}
The upper bounds of the diagonal part of the heat kernel on compact manifolds is given by (see \cite{pointinteractionsonmanifolds2}) 
\begin{eqnarray}
    K_t(a,a) < \left( \frac{1}{V(\mathcal{M})} + \frac{C}{t} \right) \;,
\end{eqnarray}
where $V(\mathcal{M})$ is the volume of $\mathcal{M}$ and $C$ is a constant depending only on the geometric properties of $\mathcal{M}$. This upper bound shows that $  \sum_{n=0 }^\infty \frac{|\phi_n(a)|^2}{(E_n - E)(E_n+\mu^2)}$ is convergent. 

For the complex values of $E=E_R + i E_I$ with $E_I \neq 0$, we can also show that the above series is absolutely convergent since
\begin{eqnarray}
\sum_{n=0}^{\infty} \frac{|\phi_n(a)|^2 }{|(E_n-E_R -iE_I)|(E_n + \mu^2 )} & = &  \sum_{n=0}^{\infty} \frac{|\phi_n(a)|^2 }{\sqrt{(E_n-E_R)^2 + E_I^2 }(E_n + \mu^2 )} \nonumber\\
& \leq &   \sqrt{2} \sum_{n=0}^{\infty} \frac{|\phi_n(a)|^2 }{(|E_n-E_R| +| E_I|) (E_n + \mu^2 )} \\
& \leq &  \sum_{n> M}^{\infty} {\sqrt{2} |\phi_n(a)|^2 \over (E_n+\mu^2)(|E_n-E_R|+|E_I|)}+  \frac{\sqrt{2}}{|E_I|} \sum_{n \leq M } \frac{|\phi_n(a)|^2 }{(E_n + \mu^2 )} < \infty \; \nonumber
\end{eqnarray}
choosing an $M$ such that for $n>M$ we have $E_n>2|E_R|$ we see that all the terms become finite.
Hence in all possible cases we have finite expressions.

One can similarly show that the same sum is convergent for Schr\"{o}dinger operators with locally integrable potentials using the upper bound for the heat kernel 
$K_t(a,a) \leq \frac{1}{4\pi t}$, given in example 2.1.9 in \cite{Davies}.

\section*{Appendix B: A Single Parameter Dependence for the Renormalized Theory}
It is interesting to note that the renormalized theory, actually have a {\it single parameter dependence} even though it seems to allow for two parameters $\alpha_R$ and $\mu^2$ in our formulae. This is of course well-known from the theory of self-adjoint extensions, yet we would like to verify explicitly that the zeros of $\Phi(E)$ will not change if we allow for a special dependence of $\alpha_R$ on $\mu$. 
Let us consider a  compact two dimensional manifold, and  demand  the function $\Phi(E)$ to be invariant under the change of $\mu$ by adjusting $\alpha_R$ accordingly (we can also demand the zeros of $\Phi$ to remain invariant as we change $\mu$, being a meromorphic function with fixed pole structure, we can then see that the function $\Phi$ remains invariant).  This will lead to a differential equation;
\begin{eqnarray}
  \mu  { \partial \over \partial \mu} {1\over \alpha_R}&=&2\mu^2 \sum_{n=0}^{\infty} {|\phi_n(a)|^2\over (E_n-E)(E_n+\mu^2)}-(E+\mu^2) \sum_{n=0}^{\infty} {2\mu^2|\phi_n(a)|^2\over (E_n-E)(E_n+\mu^2)^2}=2\mu^2 \sum_{n=0}^{\infty} {{ |\phi_n(a)|^2\over (E_n+\mu^2)^2}}\nonumber\\
    &=& 2\mu^2 \int_0^\infty dt \; t \; K_t(a,a) e^{-\mu^2 t} 
    \end{eqnarray}
which clearly shows that the left side has no dependence on $E$ and a well-defined function of $\mu^2$, hence can be integrated to find $\alpha_R$ as a function of $\mu$, it is a kind of flow equation. Incidentally, this flow equation for the beta function was previously observed in \cite{ErmanTurgut} for a restricted region of $E$ values.
In the same way, our relativistic model similarly depends on a single parameter, if we adjust the flow of $\alpha_R$ accordingly. Indeed if we compute 
\begin{equation}
    {\partial \over \partial \mu }{1\over \alpha_R}=\sum_{n=0}^{\infty} { |f_n(a)|^2\over \omega_n (\omega_n+\mu)^2}.
\end{equation}
The expression on the right is a positive finite quantity and we observe a similar flow expression, demonstrating single parameter dependence of the theory. (Its relation to the heat kernel can be found if desired using the method in Appendix C)

\section*{Appendix C: The Behavior of the Sum $\sum_{n=0}^{\infty} { (E+\mu^2)|\phi_n(a)|^2\over (E_n-E)(E_n+\mu^2)}$ as $E\to -\infty$}
In order to show that 
\begin{equation}
     \sum_{n=0}^{\infty} { (E+\mu^2)|\phi_n(a)|^2\over (E_n-E)(E_n+\mu^2)}\to -\infty \ {\rm as} \ \ E\to -\infty
\end{equation}
on a compact manifold, we  resort to a well known inequality for the heat kernel due to Cheeger and Yau \cite{CheegerYau}
\begin{eqnarray}
    K_t(x,y) \geq K_{t}^{\kappa}(d_g(x,y)) \;,
\end{eqnarray}
under the assumption that the Ricci curvature as a 2-form is bounded from below by $-\kappa g(.,.)$, which is often a reasonable condition. Here $K_{t}^{\kappa}$ is the heat kernel of the simply connected complete two dimensional manifold of constant sectional curvature and $d_g(x,y)$ is the geodesic distance between $x$ and $y$ on the manifold. In particular, we can choose $K_{t}^{\kappa}(x,y)$ as the heat kernel of the two dimensional hyperbolic manifold. 
Thanks to the work of Davis-Moundavalis \cite{DaviesMandouvalos}, there is a lower bound on the hyperbolic two dimensional manifold for the diagonal heat kernel, 
$C_{DM} {e^{-\kappa t/4}\over t (1+\kappa t)^{1/2}}$ where $C_{DM}$ is a constant, depending on dimension. This bound can even be simplified to (see \cite{ErmanMalkocTurgut})
\begin{equation}
    K_t(a,a)\geq C_{DM}{e^{-3\kappa t/4}\over t} \;.
\end{equation}
Using the Feynman parametrization for the term $1/(E_n-E)(E_n+\mu^2)$
\begin{eqnarray}
 \frac{1}{A B}=\int_{0}^{1} \frac{d u}{(A u + (1-u)B)^2} \;,   
\end{eqnarray}
and exponentiating the denominator via
\begin{equation}
    \frac{1}{(A u + (1-u)B)^2}=\int_{0}^{\infty} e^{-t \left(A u + (1-u)B\right)} \; t \; d t
\end{equation}
and then using the eigenfunction expansion of the heat kernel (\ref{eigenfunctionexpansionofheatkernel}), we can find the following upper bound for the sum  
\begin{eqnarray}
    \sum_{n=0}^{\infty} {(|E|-\mu^2)|\phi_n(a)|^2\over (E_n+|E|)(E_n+\mu^2)} & = &  \sum_{n=0}^{\infty}  (|E|-\mu^2)\int_0^1 {|\phi_n(a)|^2 du \over [ E_n+ u|E|+(1-u)\mu^2]^2} \nonumber \\ & = &  (|E|-\mu^2)\int_0^1 du \int_0^\infty dt \; t \; K_t(a,a)e^{-(t(1-u)\mu^2+u|E|)}\nonumber\\
    &\geq&  C_{DM} (|E|-\mu^2) \int_0^1 {du\over 3\kappa/4 +(1-u)\mu^2 +u|E|} \nonumber \\ & = &  C_{DM}\ln\Big({|E|+3\kappa/4\over \mu^2+3\kappa/4}\Big),
\end{eqnarray}
which demonstrates that the limit goes to $\infty$, hence there is always a solution for $1/\alpha_R<0$.
There are other bounds  with less restrictive conditions but they are more technical, the result does not change.
In a similar spirit, we can use lower bounds for the heat kernel associated with the uniformly elliptic Schr\"{o}dinger operators on ${\mathbb R}^2$ (see, e.g., theorem 3.3.4 in \cite{Davies}). If we have a nonsingular potential (strictly speaking $V$ belongs to Kato class, which also guarantees the self-adjointness of the Schr\"{o}dinger operator), there is a natural lower bound for the diagonal heat kernel given by $C/t$, so we get the same result.

For our semirelativistic model, we can show again that the expression for $\Phi(E)$ when we  choose $E<\omega_0$ always has a root between $-\mu^2$ and $\omega_0$ for positive $\alpha_R$. For internal consistency we do not want our model to develop zeroes below $-m$, yet, in any case we should show that for $\alpha_R<0$ there are solutions. To achieve this, we should verify  that $\Phi(E)$ is monotonically increasing from $-\infty$ to $\infty$ as we change $E$ from $-\infty$ to $\omega_0^-$.
It is easy to see that $\Phi(E)$ is monotone increasing in $E$ as before, in   between the poles,  however it takes some work to show  that it goes to $-\infty$ as $E\to -\infty$.

To this aim, we estimate the sum  from below assuming $E<0$, 
\begin{equation}
    \sum_{n=0}^{\infty} {|f_n(a)|^2\over \omega_n (\omega_n+\mu^2)(\omega_n+|E|)} > \sum_{n=0}^{\infty} {|f_n(a)|^2\over  (\omega_n+\mu^2)^2(\omega_n+|E|)} \;.
\end{equation}
Using the Feynman parametrization
\begin{equation}
    {1\over (\omega_n+\mu^2)^2 (\omega_n+|E|)} = 2 \int_0^1 {u \over \big[\omega_n+ u\mu^2 + (1-u)|E|\big]^3} \; d u \;,
\end{equation}
and exponentiating the denominator
\begin{equation}
    {1\over \big[\omega_n+ u\mu^2 + (1-u)|E|\big]^3} = \int_0^\infty s^2  \; e^{-s \big[\omega_n +u \mu^2 +(1-u)|E|\big]} \; ds,
\end{equation}
and then using the subordination identity (\ref{subordination}), we can express the sum in terms of the heat kernel thanks to the eigenfunction expansion of the heat kernel (\ref{eigenfunctionexpansionofheatkernel}). Then again applying the Cheeger-Yau bound for the heat kernel, assuming $R(.,.)> -\kappa g(.,.)$, we finally get an lower bound for the sum
\begin{equation}
    \sum_{n=0}^{\infty} {e^{-\omega_n^{2} t}|f_n(a)|^2}=e^{-m^2t} K_t(a,a)> C_{DM} {e^{-(m^2+3\kappa/4)t}\over t}.
\end{equation}
Plugging these results back into the original expressions, we find
\begin{eqnarray}
    & &    \sum_{n=0}^{\infty} {|f_n(a)|^2\over \omega_n (\omega_n +\mu^2)(\omega_n +|E|)} \nonumber \\ & & > \frac{1}{\sqrt{\pi}} \int_0^1 u \;  \Bigg[\int_0^\infty s^3  \Bigg(\int_0^\infty { e^{-(m^2+3\kappa/4)t-s^2/4t}\over t^{5/2}} \; dt \Bigg)  e^{-s\big(u \mu^2 +(1-u)|E|\big)}\; ds \Bigg]\;  d u \;. \label{modifiedbesselinequality}
\end{eqnarray}
We recognize that the integral with respect to $t$ is the integral representation of the modified Bessel function \cite{GradRyznik}
\begin{equation}
    K_{3/2}(z)= {1\over 2}\left({z\over 2}\right)^{3/2}\int_0^\infty {e^{-t-z^2/4t}\over t^{3/2+1}} \; dt \;, 
\end{equation}
and the modified Bessel function of order half-integer can be further simplified to the function  
\begin{equation}
    K_{3/2}(z)=\sqrt{2\over \pi z} e^{-z} + \sqrt{2\over \pi} \frac{e^{-z}}{z^{3/2}} \; . \label{besselK32}\end{equation}
It would be sufficient to concentrate on the second term in equation (\ref{besselK32}) for our purpose. Here $z=s \sqrt{m^2 + \frac{3 \kappa}{4}}$. This term after the integration with respect to $s$ in (\ref{modifiedbesselinequality}) becomes
\begin{eqnarray}
   & &  (|E|-\mu^2) \sum_{n=0}^{\infty} {|f_n(a)|^2\over \omega_n (\omega_n+\mu^2)(\omega_n+|E|)} \nonumber \\ & & > {\rm Cst} \; (|E|-\mu^2) \int_0^1 {u \over ((m^2+3\kappa/4)^{1/2} +u\mu^2 + (1-u)|E|)} \; d u \nonumber 
   \\ & & > {\rm Cst} \; (|E|-\mu^2) \int_0^1 {u \over ((m^2+3\kappa/4)^{1/2} +\mu^2 + (1-u)|E|)} \; d u \nonumber 
   \\ & & = {\rm Cst} \; (|E|-\mu^2) \int_0^1 {(1-u) \over ((m^2+3\kappa/4)^{1/2} +\mu^2 + u|E|)} \; d u \nonumber
   \\ & & = {\rm Cst} \frac{|E|-\mu^2}{|E|}\bigg( \frac{|E|+(m^2+3\kappa/4)^{1/2}+\mu^2}{|E|}\ln {|E|+[m^2+3\kappa/4]^{1/2} + \mu^2 \over \mu^2 + [m^2+3\kappa/4]^{1/2}} -1 \bigg) \;,
\end{eqnarray}
hence proving again that the sum goes to $\infty$ as $|E|\to \infty$. Here the first term in (\ref{besselK32}) gives positive and regular contribution to the sum, which do not change our conclusion.

\section*{Appendix D: Degenerate Case}

Here we consider the possibility of degeneracy in the eigenstates, {\it for clarity of our presentation,  we only look into the two-dimensional compact manifold case}. More general cases can be done using exactly the same idea. In this case assume that there is an eigenvalue $E_k$ which is $d$-fold degenerate (for a manifold it is known to be finite for every state, and for its generalization to other cases, it  is  always assumed to be a finite degeneracy). We choose a set of orthonormal vectors $\phi_k^{\alpha}$ to span this space. 
In our Green's function formula, assuming $E$ is around $E_k$, we can explicitly write the leading terms with $E_k$ poles;
\begin{eqnarray}
 & & \hskip-1cm    G(x,y|E)= g(x,y|E) + \sum_{\alpha=1}^{d} { \phi_k^{\alpha}(x)\overline{\phi_k^{\alpha}(y)}\over E_k-E}+ \Bigg[ g(x, a|E)\nonumber\\
    & & \hskip-1cm + \sum_{\alpha=1}^{d} { \phi_k^{\alpha}(x) \overline{\phi_k^\alpha(a)}\over (E_k-E)}\Bigg]{1\over {1\over \alpha_R}- k(a,a|E)- \sum_{\alpha=1}^{d} { \phi_k^\alpha(a) \overline{\phi_k^\alpha(a)}(E+\mu)\over (E_k-E)(E_k+\mu)}}\Bigg[ g(a, y|E) + \sum_{\beta=1}^{d} { \phi_k^\beta(a) \overline{\phi_k^\beta(y)}\over E_k-E}\Bigg] \;,
\end{eqnarray}
here  $\alpha_R$ refers to  the renormalized coupling constant and 
\begin{eqnarray}
    g(x,y|E) &:=& \sum_{n\neq k}^{\infty}  { \phi_n(x) \overline{\phi_n(y)} \over E_n-E} \nonumber \\ k(x,y|E) & := & \sum_{n\neq k}^{\infty} {\phi_n(x) \overline{\phi_n(y)}(E+\mu)\over (E_n-E)(E_n+\mu)} \;.
\end{eqnarray}
We now massage the pole term in the denominator,
\begin{equation}
     { \phi_k^\alpha(a)\overline{\phi_k^\alpha(a)}(E+\mu)\over (E_k-E)(E_k+\mu)}={ \phi_k^\alpha(a) \overline{\phi_k^\alpha(a)}(E+\mu-E_k+E_k)\over (E_k-E)(E_k+\mu)}= { \phi_k^\alpha(a)\overline{\phi_k^\alpha(a)}\over (E_k-E)}-{ \phi_k^\alpha(a)\overline{\phi_k^\alpha(a)}\over (E_k+\mu)}
,\end{equation}
where the pole term is isolated with other part being regular.
In our expression, the crucial part is the product of the two poles terms and then we move one of $E_k-E$ terms into the denominator to get a leading order term plus a power series in $E_k-E$:
\begin{eqnarray}
   & & \hskip-2cm  G(x,y|E)= f(x,y|E) + \sum_{\alpha=1}^{d} { \phi_k^\alpha(x)\overline{\phi_k^\alpha(y)}\over E_k-E}\nonumber\\
    & &  \hskip-1cm + {1\over E_k-E}\sum_{\alpha=1}^{d} { \phi_k^\alpha(x)\overline{\phi_k^\alpha(a)}}{1\over -\sum_\gamma { \phi_k^\gamma(a)\overline{\phi_k^\gamma(a)}}+ (E-E_k)h(a,a|E)} \sum_{\beta=1}^{d} { \phi_k^\beta(a)\overline{\phi_k^\beta(y)}} \;,
\end{eqnarray}
where
\begin{eqnarray}
   & & \hskip-2cm  f(x,y|E) = g(x,y|E) + \frac{g(x,a|E)g(a,y|E)}{{1\over \alpha_R}- k(a,a|E)- \sum_{\alpha=1}^{d} { \phi_k^\alpha(a) \overline{\phi_k^\alpha(a)}(E+\mu)\over (E_k-E)(E_k+\mu)}} \nonumber \\ & & +  \frac{g(x,a|E)}{(E_k-E)\left({1\over \alpha_R}- k(a,a|E)\right)- \sum_{\alpha=1}^{d} { \phi_k^\alpha(a) \overline{\phi_k^\alpha(a)}(E+\mu)\over (E_k+\mu)}} \sum_{\beta=1}^{d} {\phi_k^\beta(a)\overline{\phi_k^\beta(y)}} \nonumber \\ & & + \frac{g(a,y|E)}{(E_k-E)\left({1\over \alpha_R}- k(a,a|E)\right)- \sum_{\alpha=1}^{d} { \phi_k^\alpha(a) \overline{\phi_k^\alpha(a)}(E+\mu)\over (E_k+\mu)}} \sum_{\beta=1}^{d} {\phi_k^\beta(x)\overline{\phi_k^\beta(a)}}  \;,
\end{eqnarray}
and 
\begin{eqnarray}
    h(a,a|E)= k(a,a|E) - \frac{1}{\alpha_R} - \sum_{\alpha=1}^{d} \frac{|\phi_{k}^{\alpha}(a)|^2}{E_k+\mu^2} \;.
\end{eqnarray}
Note that here, $f(x,y|E)$ is regular around $E_k$ and in the denominator, we have a leading finite term and then a regular function multiplied with $E-E_k$, which can be made sufficiently small. 
We now assume that {\it at least one of the wave functions} $\phi_k^\alpha(a)$ is nonzero. Then, we can expand this combination in the denominator and combine them with the regular function $f$ and call it new regular function by $l$. Hence, we find 
\begin{eqnarray}
    G(x,y|E)&=& l(x,y|E) + \sum_{\alpha=1}^{d} {\phi_k^\alpha(x)\overline{\phi_k^\alpha(y)}\over E_k-E}\nonumber\\
    &\ & \ - {1\over E_k-E}\sum_{\alpha=1}^{d} {\phi_k^\alpha(x)\overline{\phi_k^\alpha(a)}}{1\over \sum_{\gamma=1}^{d} {\phi_k^\gamma(a)\overline{\phi_k^\gamma(a)}}}\sum_{\beta=1}^{d} {\phi_k^\beta(a)\overline{\phi_k^\beta(y)}} \;.
\end{eqnarray}
Let us define 
$A_{\alpha\beta}=\overline{\phi_k^\alpha(a)}\phi_k^\beta(a)$, which 
is a manifestly Hermitian matrix, with trace Tr$A=\sum_{\alpha=1}^{d} \overline{\phi_k^\alpha(a)}\phi_k^\alpha(a)=\sum_{\alpha=1}^{d} |\phi_k^\alpha(a)|^2 $
Moreover, if we think of $\phi^\alpha_k(a)$ as components of a vector in $d$ dimensions, $A$ has {\it only one non-zero eigenvalue}, which is Tr$A$, with eigenvector $\phi^\gamma_k(a)$. As $A_{\alpha\beta} v_\beta=0$ for all the vectors $v$ orthogonal to single vector $\phi_k^\beta(a)$, we have $d-1$ zero eigenvalues.
$A_{\alpha\beta}$ being a Hermitian matrix, can be diagonalized by a unitary $U_{\alpha\beta}$ matrix:
\begin{equation}
    U^\dagger A U={\rm diag} (0, 0 , \cdots, {\rm Tr} (A)) \;.
\end{equation}
Consider $\sum_{\beta=1}^{d} U_{\alpha\beta}\phi^\beta_k(x)=\psi_k^\alpha(x)$, $U$ being unitary in the $d$ dimensional subspace, we get orthogonal vectors again in $E_k$ subspace.  
Note that the combination with the pole structure goes over to 
\begin{equation}
\sum_{\alpha=1}^{d} {\phi_k^\alpha(x)\overline{\phi_k^\alpha(y)}\over E_k-E}=\sum_{\alpha=1}^{d} {\psi_k^\alpha(x) \overline{\psi^\alpha_k(y)}\over E_k-E} \;.
\end{equation}
The other part instead
\begin{eqnarray}
   &\ &  - {1\over E_k-E}\sum_{\alpha=1}^{d} { \phi_k^\alpha(x)\overline{\phi_k^\alpha(a)}}{1\over \sum_{\gamma=1}^{d} { \phi_k^\gamma(a)\overline{\phi_k^\gamma(a)}}}\sum_{\beta=1}^{d} {\phi_k^\beta(a)\overline{\phi_k^\beta(y)}}=\nonumber\\
   &\ & =-{1\over E_k-E} {1\over {\rm Tr}(A)} \sum_{\alpha=1}^{d} \phi_k^\alpha(x)A_{\alpha\beta} \overline{\phi^\beta_k(y)}=\nonumber\\
   &\ & =-{1\over E_k-E} {1\over {\rm Tr}(A)} \Big[ \psi^1_k(x) \cdot 0 \cdot \overline{\psi^1_k(y)}+ ...+\psi^{d-1}_k(x)\cdot 0 \cdot \overline{\psi^{d-1}_k(y)} + \psi^d_k(x){\rm Tr}(A) \overline{\psi_k^d(y)}\Big]\nonumber\\
   &\ & =-{1\over E_k-E} \psi^d_k(x)\overline{\psi^d_k(y)} \;.
\end{eqnarray}
Thus we see that there is a single cancellation, the wave function that corresponds to Tr$A$, $\psi^d_k(x)$ terms cancel, and we end up with a $d-1$ dimensional subspace with a pole $E_k$. Degeneracy is lifted by removing  one of the properly chosen eigenstates. 
As one can see, the crux of the argument is this diagonalization process, and it applies equally well in all other cases when there is degeneracy.

\end{document}